\documentclass[11pt]{extarticle}

\usepackage{soul}
\usepackage{url}
\usepackage[hidelinks,breaklinks]{hyperref}
\usepackage[utf8]{inputenc}
\usepackage[small]{caption}
\usepackage{graphicx}
\usepackage{booktabs}

\usepackage{amsmath,multirow}

\usepackage{enumerate}
\usepackage{etoolbox}
\usepackage{subcaption}
\usepackage{xcolor,xfrac}
\usepackage{natbib}
\usepackage{hyperref}
\usepackage[english]{babel}
\usepackage{graphicx,charter}
\usepackage{authblk}
\usepackage{framed}
\usepackage[normalem]{ulem}
\usepackage{amsmath,multirow}
\usepackage{amsthm,thm-restate,thmtools}
\usepackage{tikz}
\usepackage{amssymb}
\usepackage{amsfonts}
\usepackage{enumerate}
\usepackage[T1]{fontenc}
\usepackage[utf8]{inputenc}
\usepackage[top=1 in,bottom=1in, left=1 in, right=1 in]{geometry}
\usepackage{url}
\usepackage{enumitem,natbib}
\usepackage{tabularx}
\usepackage{algorithm}
\usepackage{algorithmic}
\usepackage{chngcntr}

\hypersetup{
     colorlinks   = true,
     linkcolor    = red, 
     urlcolor     = teal, 
	 citecolor    = blue 
}

\usepackage{amssymb}
\usepackage{mathtools}
\usepackage{thmtools,thm-restate}
\usepackage{etoolbox}

\newcommand{\citepapp}[1]{\cite{#1}}
\newcommand{\citetapp}[1]{\citeauthor{#1}~[\citeyear{#1}]}

\newcommand{\appplaceholder}{}

\usepackage[T1]{fontenc}

\usepackage{nameref}

\usepackage{pifont}
\usepackage{xcolor}
\usepackage{colortbl}
\usepackage{hyperref}[backref=page]
\usepackage{subcaption}

\usepackage[capitalise,noabbrev]{cleveref}
\Crefname{remark}{Remark}{Remarks}
\Crefname{rmk}{Remark}{Remarks}
\Crefname{dfn}{Definition}{Definitions}
\Crefname{thm}{Theorem}{Theorems}
\Crefname{cor}{Corollary}{Corollaries}
\Crefname{lem}{Lemma}{Lemmas}
\Crefname{examplex}{Example}{Examples}
\Crefname{prop}{Proposition}{Propositions}
\Crefname{claim}{Claim}{Claim}
\Crefname{claim_app_EF}{Claim}{Claim}
\Crefname{claim_app_EFX}{Claim}{Claim}
\Crefname{claim_app_EF1}{Claim}{Claim}

\usepackage{tikz}

\newtheorem{theorem}{Theorem}
\newtheorem{example}{Example}

\newtheorem{proposition}{Proposition}
\newtheorem{corollary}{Corollary}
\newtheorem{claim}{Claim}
\newtheorem{remark}{Remark}
\newtheorem{definition}{Definition}
\newtheorem{claim_app_EF}{Claim}
\newtheorem{claim_app_EFX}{Claim}
\newtheorem{claim_app_EF1}{Claim}

\newcommand{\MMS}{\textrm{\textup{MMS}}}

\newcommand{\EFX}{\textrm{\textup{EFX}}}
\newcommand{\EF}[1]{\ifstrempty{#1}{\textrm{\textup{EF}}}{\textrm{\textup{EF{$#1$}}}}}

\newcommand{\PO}{\textup{PO}}

\newcolumntype{P}[1]{>{\centering\arraybackslash}p{#1}}
\newcolumntype{M}[1]{>{\centering\arraybackslash}m{#1}}

\title{Almost Envy-Freeness 

under Weakly Lexicographic Preferences}

\author{Hadi Hosseini}
\author{Aghaheybat Mammadov}
\author{Tomasz Wąs}
\affil{Pennsylvania State University\vspace{0.15cm}\\
 {\normalsize \{hadi, mammadovagha, twas\}@psu.edu}}
\date{}

\begin{document}
\maketitle
\begin{abstract}
\noindent
In fair division of indivisible items, domain restriction has played a key role in escaping from negative results and providing structural insights into the computational and axiomatic boundaries of fairness. 
One notable subdomain of additive preferences, the lexicographic domain, has yielded several positive results in dealing with goods, chores, and mixtures thereof.
However, the majority of work within this domain primarily consider strict linear orders over items, which do not allow the modeling of more expressive preferences that contain indifferences (ties).
We investigate the most prominent fairness notions of envy-freeness up to any (\EFX{}) or some (\EF{1}) item under weakly lexicographic preferences. 
For the goods-only setting, we develop an algorithm that can be customized to guarantee \EF{1}, \EFX{}, maximin share (\MMS{}), or a combination thereof, along the efficiency notion of Pareto optimality (\PO{}).
From the conceptual perspective, we propose techniques such as preference graphs and potential envy that are independently of interest when dealing with ties.
Finally, we demonstrate challenges in dealing with chores and highlight key algorithmic and axiomatic differences of finding EFX solutions with the goods-only setting.
Nevertheless, we show that there is an algorithm that always returns an \EF{1} and \PO{} allocation for the chores-only instances. 

\end{abstract}

\section{Introduction}

The distribution of indivisible items is a fundamental problem in a wide array of societal, computational, and economic settings.
Over the past few decades, the field of fair division has emerged to promote fairness in designing scalable algorithms for distributing resources and tasks.
Arguably, one of the primary drivers behind recent advancements in algorithmic fairness was the focus on restricted domains (e.g., additive valuations) that enabled concise representation of preferences in complex resource allocation problems.
Domain restriction has played a key role in escaping from negative results and providing structural insights into the computational and axiomatic boundaries in computational social choice \citep{lang2018voting,HL19multiple,nguyen2020fairly,fujita2018complexity}. 
One notable subdomain of additive preferences, the lexicographic domain, has received much attention due to its succinct representation of complex preferences and its natural proxy for modeling consumer behavior \citep{gigerenzer1996reasoning,fishburn1974exceptional}.

Within fair division, the lexicographic domain has resulted in several positive results in dealing with goods \citep{hosseini2021fair}, chores \citep{ebadian2022fairly,hosseini2022fairly}, and mixtures thereof \citep{hosseini2023fairly}.
These results primarily consider strict linear orderings over the items, which do not allow for the modeling of more expressive preferences that contain ties (aka weak preferences).
Yet, individuals are often indifferent between sets of items and tend to group those into `equivalence classes'.
For example, in conference paper reviewing, one may strictly prefer papers from computational social choice to those from computer vision, but the same reviewer may be indifferent between papers that cover topics in fair division.

The introduction of weak preferences poses several intriguing algorithmic challenges, particularly when dealing with economic efficiency; as such ties have been in the center of attention in a large body of work, for example, in 
probabilistic assignments with ordinal or lexicographic preferences \citep{katta2006solution,saban2014note,aziz2015fair}, 
mechanism design for object allocation \citep{bogomolnaia2005strategy,krysta2014size},
assigning papers to reviewers \citep{garg2010assigning}, and
Shapley-Scarf housing markets \citep{klaus2023core,saban2013house,jaramillo2012difference,aziz2012housing}.

Along this line, we focus on expanding the existing computational and axiomatic results in fair division to the more expressive preference class that allows for `mild continuity' in preferences over items.
We investigate the most prominent fairness notion, envy-freeness \EF{}, and its relaxations envy-freeness up to \textit{any} item (\EFX{}) \citep{caragiannis2019unreasonable} and envy-freeness up to \textit{some} item (\EF{1}) \citep{lipton2004approximately,budish2011combinatorial}, along with an efficiency requirement of Pareto optimality (\PO{}).
\EFX{} requires that a pairwise envy between the two agents is eliminated by the removal of any good (or chore) from the envied (envious) agent's bundle, while \EF{1} relaxes this requirement to the removal of some single item from the bundle of one of the two agents.
Furthermore, we consider fairness notion of \emph{maximin share} (\MMS{}) \citep{budish2011combinatorial},
which states that an agent should receive a bundle at least as preferable
as the best possible bundle it can get
in a partition 
from which it receives the worst bundle.

\subsection{Contributions}
We study fair and efficient allocations 
in goods-only and chores-only instances 
under weakly lexicographic preferences. 
We present computational and axiomatic results 
regarding well-established fairness notions of
\EF{} and its prominent relaxations, 
\EFX{} and \EF{1}, and \MMS{}, 
as well as conceptual techniques which 
hold independent significance when dealing with ties.

\paragraph{Goods.} Our first result shows that 
with the introduction of ties, deciding if 
an \EF{} allocation of goods exists becomes NP-complete, 
even when agents have at most two indifference classes 
(\cref{thm:goods:ef:hardness}). 
Given the nonexistence and computational hardness for \EF{}, 
we develop a customizable algorithm that
finds a \PO{} allocation satisfying
an \EF{} relaxation of \EF{1}, \MMS{}, \EFX{}, 
or a combination thereof,
depending on the chosen parameters (\cref{thm:goods:efx}).

This way, we are able to pinpoint to an exact place in our algorithm that is responsible for guaranteeing fairness.

\paragraph{Chores.} Under additive preferences, 
there has been much less progress for chores, and as such, 
the existence of \EF{1} and \PO{} is still open. We 
prove the existence and computation of such an allocation 
under weakly lexicographic preferences (\cref{thm:chores:ef1}) as the \textit{largest} known subclass of additive valuations for which \EF{1} and \PO{} can be guaranteed.
Moreover, we illustrate the challenges 
involved in fairly allocating chores 
(\cref{ex:chores:not_easy1,ex:chores:not_easy2}), and show that in the chores-only instances, \EFX{} implies \MMS{} (\cref{prop:chores:efx:implies:mms}), which surprisingly stands in contrast to the case of goods (\cref{app:efx:mms:relation}). 

\paragraph{Conceptual Techniques.} 
From the technical perspective, we develop 
novel techniques such as 
preference graphs and potential envy (\cref{def:potential_envy}
and \cref{fig:goods:ex:g_pref}) for handling indifferences when the goal is to guarantee fairness along with efficiency in this setting.

\section{Related Work}

It is well-known that for goods-only instances with additive preferences an \EF{1} allocation 
can be efficiently computed~\citep{lipton2004approximately}. 
Moreover, an allocation satisfying \EF{1} and \PO{} always exists~\citep{caragiannis2019unreasonable} and can be computed in
pseudo-polynomial time~\citep{barman2018finding}. 
Nevertheless, the existence of 
\EFX{} allocations remains unresolved for this general setting. 
To gain insights into the problem,
a wide variety of restricted domains have been considered. 
In particular, it has been shown that 
an \EFX{} allocation always exists
for instances with identical valuations~\citep{PR20almost},
under submodular utilities 
with binary marginals~\citep{BEF21fair},
and for additive preferences with interval valuations or 
at most two distinct values~\citep{ABF+21maximum,garg2021computing}.
However, 
\EFX{} is not compatible with \PO{} 
under submodular valuations~\citep{PR20almost}. 
Nevertheless, under strict lexicographic preferences, 
an \EFX{} and \PO{}
allocation always exists and can be computed 
in polynomial time~\citep{hosseini2021fair}. 

In studying possible fairness, \citet{aziz2023possible} proposed
an algorithm for finding an \EFX{}, \MMS{}, and \PO{} allocation
under weakly lexicographic preferences.
Our algorithm achieving the same guarantees
has been developed concurrently and independently.
Moreover, there are significant differences between these algorithms,
both in terms of their outputs
(see \cref{ex:goods:efx})
as well as construction.
The main difference on the conceptual level
is that our algorithm can be easily alternated to cover different notions of fairness, allowing for better understanding of where the particular fairness guarantee is invoked.
For more detailed comparison see \cref{app:comparison}.

Analogous problems for chores-only instances seem much more complex.
\EF{1} allocations can still be efficiently computed~\citep{bhaskar2021approximate,aziz2022fair}, however the existence of \EF{1} and \PO{} allocations remains an important open problem, except for a few restricted domains~\citep{ebadian2022fairly,garg2022fair,aziz2023fairtypes}.
Despite much efforts, the progress for \EFX{} is also limited. Recent works established its existence for 
ordered instances~\citep{li2022almost}, 
generalized binary 
preferences~\citep{camacho2023generalized},
strictly lexicographic preferences~\citep{hosseini2022fairly},
and 
instances with two types of chores~\citep{aziz2023fairtypes}.

The next most studied fairness notion is arguably \MMS{}.
It may fail to exist for both 
goods-only~\citep{kurokawa2018fair} and 
chores-only~\citep{aziz2017algorithms}
instances under additive preferences. 
Therefore, both multiplicative approximations~\citep{aziz2017algorithms,ghodsi2018fair,garg2020improved} and
ordinal ones ~\citep{babaioff2019fair,hosseini2022ordinal} 
have been proposed and analyzed. 
\citep{ebadian2022fairly} showed, among others, that an \MMS{} and \PO{} allocation can be efficiently computed in every goods-only or chores-only weakly lexicographic instance.

\section{Preliminaries}

For every $k \in \mathbb{N}$, let $[k]=\{1,\dots,k\}$.
A weakly lexicographic instance
is defined as a triple $(N,M,\rhd)$,
in which $N = [n]$ is a set of \emph{agents},
$M$ a set of \emph{items},
and $\rhd = (\rhd_i)_{i \in N}$ a family of
weak linear orders
over $M$,
based on which we will define the preferences of the agents.
We will consider two types of instances:
in a \emph{goods-only} instance,
we say that
$M = \{g_1,\dots,g_m\}$ is a set of \emph{goods},
while in a \emph{chores-only} instance,
it is a set of \emph{chores}, $M = \{c_1,\dots,c_m\}$.

\paragraph{Indifference classes.}
\noindent
For each agent $i \in N$, a weak linear order $\rhd_i$
partitions the set $M$ into $k_i$ indifference classes,
$\rhd_i(1), \dots, \rhd_i(k_i)$,
for some integer $k_i \le |M|$.
Intuitively, the agent is indifferent between
two items from the same class,
but for every $k \in [k_i]$,
a single element of the $k$-th class, i.e., $\rhd_i(k)$
cannot be compensated by any number of items
from the later classes, i.e., $\rhd_i(k')$ for $k'> k$.
In examples,
we will define the weak linear orders
by explicitly writing the indifference classes
divided by a triangle sign, $\rhd$, to denote relation between them, e.g.,
\begin{equation}
    \label{eq:ex}
    1: \ g_1 \rhd \{g_2, g_3\} \rhd g_4.
\end{equation}

\paragraph{Preferences.}
For every subset $X \subseteq M$,
and indifference class $k \in [k_i]$,
let $\rhd_i(k, X) = \rhd_i(k) \cap X$.
Then, we will denote a \emph{score} vector,
$s_i(X) = (s_i(1,X),\dots,s_i(k_i,X))$,
in which for each $k \in [k_i]$
its $k$-th coordinate is defined as
$s_i(k, X) = |\rhd_i(k, X)|$, for goods-only instances,
and $s_i(k, X) = - |\rhd_i(k, X)|$, for chores-only instances.
For $X = M$, we will write simply $s_i(k)$ for brevity.
For two sets of items, $X, Y \subseteq M$,
agent $i$ \emph{strictly prefers} set $X$ over $Y$,
i.e., $X \succ_i Y$,
if the score of $X$
lexicographically dominates the score of $Y$,
i.e., $s_i(X) >_{lex} s_i(Y)$.
In other words,
there exists $\bar{k} \in [k_i]$ such that
$s_i(\bar{k},X) > s_i(\bar{k},Y)$, and
$s_i(k,X) = s_i(k,Y)$, for every $k < \bar{k}$.

For example, agent 1 from \eqref{eq:ex}
prefers $X = \{g_1,g_4\}$ over $Y = \{g_2,g_3,g_4\}$ as
\(
    s_i(X) = (1,0,1) >_{lex} (0,2,1) = s_i(Y).
\)
If $s_i(X) = s_i(Y)$,
then agent $i$ is \emph{indifferent} between the sets,
i.e., $X \sim_i Y$.
Finally, agent $i$ \emph{weakly prefers} $X$ over $Y$,
denoted by $X \succeq_i Y$, if either $X \succ_i Y$ or $X \sim_i Y$.
For singleton sets,
we will skip the brackets in these relations, e.g.,
we will write $g_1 \succ_1 g_2 \sim_1 g_3$.

\paragraph{Allocations.}
An allocation $A = (A_1,\dots,A_n)$ is an $n$ partition of $M$,
such that the \emph{bundles} of agents are disjoint, i.e.,
$A_i \cap A_j = \emptyset$, for every $i,j \in N$.
If all items are allocated, i.e., $\bigcup_{i\in N} A_i = M$,
we say that $A$ is \emph{complete}.
It is \emph{partial} otherwise.
The goal is to find a complete allocation
that is fair and efficient.

\paragraph{Envy-freeness.}
An agent $j$ \emph{envies} $i$ if $A_i \succ_j A_j$.
An allocation $A$ is said to be 
\emph{envy-free} (\EF{})
if no agent envies any other.
It is \emph{envy-free up to one item} (\EF{1}),
            if for every pair of agents $i,j \in N$ such that 
            $j$ envies $i$,
            $A_j \succeq_j A_i \setminus \{g\}$
            for some $g \in A_i$ in case of 
            a goods-only instance, or
            $A_j \setminus \{c\} \succeq_j A_i $
            for some $c \in A_j$ in case of 
            a chores-only instance.
Finally, allocation $A$ is \emph{envy-free up to any item} (\EFX{}),
            if for every pair of agents $i,j \in N$,
            $A_j \succeq_j A_i \setminus \{g\}$
            for \textit{every} $g \in A_i$ in case of 
            a goods-only instance, or
            $A_j \setminus \{c\} \succeq_j A_i $
            for \textit{every} $c \in A_j$ in case of 
            a chores-only instance.

\paragraph{Maximin share.}
A \emph{maximin share} of an agent is the score vector
of the most preferred bundle it can guarantee for itself
by dividing the items into $n$ bundles and
receiving the worst one.
Formally,
\(
    \MMS{}_i = \max_{(A_1,\dots,A_n) \in \Pi^n} \min_{j \in [n]} s_i(A_j),
\)
where $\Pi^n$ is a set of all possible $n$-partitions
and $\max$ and $\min$ are determined based on lexicographic dominance.
An allocation $A$ satisfies \emph{maximin share} (\MMS{})
if for every agent $i \in N$,
it holds that
$s_i(A_i) \ge_{lex} \MMS{}_i$.

\paragraph{Pareto optimality.}
An (possibly partial) allocation $A$ \emph{Pareto dominates} allocation $B$
if $A$ assigns the same set of goods as $B$, i.e.,
$\bigcup_{i\in N}  A_i = \bigcup_{i\in N}  B_i$, and
$A_i \succeq_i B_i$ for every $i \in N$
and there exists $j \in N$ such that $A_j \succ_j B_j$.
Allocation $A$ is \emph{Pareto optimal} (\PO{}), if it is not Pareto dominated by any other allocation. 

\section{Goods}
\label{sec:goods}

Let us start with weakly lexicographic goods-only instances.

\citep{hosseini2021fair} proved that we can
decide if a strict lexicographic goods-only instance admits an \EF{} allocation
in polynomial time.
However, we show that the same is no longer true
for weakly lexicographic instances.

\begin{restatable}{theorem}{thmgoodsefhardness}\label{thm:goods:ef:hardness}
Deciding whether an \EF{} allocation exists for a given weakly lexicographic goods-only instance is NP-complete, even if every agent has at most two indifference classes.
\end{restatable}

The proof is relegated to \cref{app:ef:np-hardness}, where we also show NP-completeness of \EF{} combined with \PO{}.
We build upon similar proofs by
\citet{aziz2015fair} and \citet{hosseini2020HEF}
for binary preferences.
However, under binary preferences a good may
hold no additional value to an agent,
which is not possible in our setting
(i.e., an agent always
strictly prefers to have a good rather than not).
As a result, our construction is more intricate
and requires additional steps in the proof.

Motivated by the non-existence and computational hardness for \EF{},
we move to its most compelling relaxation: \EFX{}. 
In what follows, we will
provide an algorithm (\cref{alg:goods:efx}) that,
depending on a given requirement,
finds an allocation that satisfies \PO{} as well as
\EF{1}, \MMS{}, \EFX{}, or \MMS{} and \EFX{} simultaneously (\cref{thm:goods:efx}).

\subsection{Preference Graph and Potential Envy}

We begin by introducing additional constructions and definitions which our algorithm utilizes.

\paragraph{Preference graph.}
First,
let us introduce a \emph{preference graph},
which is a weighted complete bipartite graph, $G_{pref} = (N,M,E,\psi)$,
in which both agents and goods are vertices,
all agent-good pairs are edges $E = N \times M$,
and the weights, $\psi : E \rightarrow \mathbb{N}$,
denote in which indifference class is a given good
for a given agent, i.e., for every $(i, g) \in E$,
$\psi(i, g) = k$ such that $g \in \rhd_i(k)$.
An example of a preference graph can be found in
\cref{fig:goods:ex:g_pref}.
For every (possibly partial) allocation $A$,
we will slightly abuse the notation and treat $A$
both as a collection of subsets, $(A_1,\dots,A_n)$,
and as a subset of edges in a preference graph,
$\bigcup_{i \in N} \{ (i,g) : g \in A_i\}$,
depending on the context.
Given a preference graph $G_{pref} = (N,M,E,\psi)$
and an (possibly partial) allocation $A$,
an \emph{alternating path}, $p = (g_0, i_1, g_1, \dots, i_s, g_s)$,
is a path in $G_{pref}$ 
(possibly of length zero) such that
\begin{itemize}
    \item the goods in $p$ are pairwise distinct,
    \item $p$ alternates between edges
    that belong and do not belong to $A$,
    i.e., $(g_{r-1}, i_{r}) \in A$, for every $r \in [s]$, and
    $(i_{r}, g_r) \not \in A$, for every $r \in [s]$, and
    \item for every $r \in [s]$,
    agent $i_{r}$ weakly prefers good $g_{r}$ over $g_{r-1}$,
    i.e., $\psi(g_{r-1}, i_r) \ge \psi(i_{r}, g_{r})$.
\end{itemize}

\paragraph{Available goods.}
We say that a good $g$ is \emph{available}
if there exists
an alternating path from $g$ to
an unallocated good
(or $g$ is unallocated itself).
Otherwise $g$ is \emph{unavailable}.
If an alternating path $p = (g_0, i_1, g_1, \dots, i_s, g_s)$
ends with an unallocated good, then
allocation $A$ can be \emph{updated}
along $p$,
to obtain allocation $A'$
such that
\(
    A' =  A \setminus
        \{(g_{r-1}, i_{r}) : r \in [s]\} \cup
        \{(i_r, g_r) : r \in [s]\}.
\)
Observe that every agent $i \in N$
weakly prefers its bundle in $A'$ compared to $A$.
Note also that for a given good $g$,
we can find an alternating path starting in $g$
and ending in an unallocated good
(or conclude that there is none)
using the BFS algorithm in time $O(nm)$.

Moreover, a preference graph can be used
to efficiently check
whether an allocation is PO.
Rewriting the result of \citet{aziz2019reallocation}
using a preference graph,
we get that the allocation is PO,
if and only if,
there is no alternating path, $p$, and agent, $i$,
such that the last good in the path, $g_s$,
belongs to agent $i$,
but $i$ strictly prefers
the first good in the path over $g_s$.

\begin{theorem}{\cite[Theorem 5]{aziz2019reallocation}}
    \label{thm:po:aziz}
    Given a weakly lexicographic
    goods-only instance $(N,M,\rhd)$,
    an (possibly partial) allocation $A$ is PO,
    if and only if,
    there is no alternating path,
    $p = (g_0, i_1, g_1, \dots, i_s, g_s)$ and agent $i_0$ such that
    $(g_s, i_0) \in A$ and
    $\psi(i_0, g_0) < \psi(g_s,i_0)$.
\end{theorem} 

Next, we introduce a notion of \emph{potential envy}
in a partial allocation,
which we will use to guarantee \EFX{}.
Intuitively, potential envy towards an agent exists,
if giving all still available goods to this agent,
would result in an (actual) envy.

\begin{definition}(Potential envy)
\label{def:potential_envy}
Given an instance $(N,M,\rhd)$,
a partial allocation $A$,
and agents $i,j \in N$, agent $j$ \emph{potentially envies} agent $i$ if $A_i \cup B \succ_j A_j$, where $B$ denotes the set of all available goods.
\end{definition}

Then, a \emph{potential envy graph}
is a directed graph, $G_{envy}(U, A) = (U, E)$,
that given a subset of agents $U$ 
and a partial allocation $A$,
puts an edge $(i,j)$ from agent $i \in U$  to $j \in U \setminus \{i\}$,
if and only if, 
$i$ potentially envies $j$
(see \cref{fig:goods:ex:g_envy} for an illustration).

Finally, let us include a characterization of MMS allocations
that will be useful in proving that
our algorithm always outputs an MMS allocation.
This characterization is based on the algorithm for
establishing an MMS threshold for 
a given agent, $i$,
which was developed by \citet{ebadian2022fairly}.
The algorithm starts with
a family of $n$ empty bundles.
Then, in each step, 
it assigns an unallocated good from
the first indifference class of agent $i$ with still unallocated goods
to a bundle with lexicographically minimal score $s_i$.
After all goods are allocated,
the \MMS{} threshold is then the
lexicographically minimal score $s_i$
among the $n$ bundles.
Based on the algorithm,
we provide a numerical
formula for the threshold.

\begin{proposition}
\label{prop:goods:mms:threshold}
Given a weakly lexicographic goods-only instance $(N,M,\rhd)$
and agent $i \in N$, it holds that
$\MMS{}_i = (x_1,x_2,\dots,x_k)$,
where $x_l \!=\! \lfloor s_i(l)/r_l \rfloor$, for every $l \!\in\! [k]$, while
$r_1 \!=\! n$, and 
$r_{l+1} \!=\! r_l \!-\! (s_i(l) \!-\! r_l x_l)$,
for every $l \!\in\! [k\!-\!1]$.
\end{proposition}

\subsection{The Algorithm}

Now, let us move to the main result of this section
which is \cref{alg:goods:efx} that finds a
fair (\EF{1}, \MMS{}, or \EFX{}) allocation which satisfies \PO{}
in polynomial time for every weakly lexicographic goods-only instance.

Our algorithm allows for flexibility in regard to the fairness guarantees it obtains. Specifically, in each iteration of our algorithm, for each agent, our algorithm decides whether we want to stop giving goods to this agent or not, based on certain criteria, $F$.
The choice of criteria affects the fairness guarantee for the outcome.
In general, these criteria can be arbitrary, but we focus on three, which, with a slight abuse of terminology, we call based on the fairness notions they turn out to guarantee:
\EFX{}, \MMS{}, and the conjunction of both \EFX{}+\MMS{}.
We also analyze the baseline version of the algorithm which never differentiates agents
(i.e., $F=$ \textit{null}), 
and show that such algorithm guarantees \EF{1} and \PO{} allocations.

\begin{algorithm}[t]\small
	\caption{Finding a fair and efficient allocation of goods}
    \label{alg:goods:efx}
    \begin{algorithmic}[1]
        \REQUIRE A weakly lexicographic goods-only instance
        $\langle N, M, \rhd \rangle$, an ordering $\sigma$, and
        criteria $F$\\
        \STATE $A \leftarrow (\emptyset,\dots,\emptyset),
        \quad G_{pref} \leftarrow (N, M, E, \psi), \quad U \leftarrow N$

        \WHILE{there is an unallocated good}
            \STATE Take $i_0 \in U$ (or $i_0 \in N$, if $U = \emptyset$) s.t. $|A_{i_0}|$ 
            is minimal
            (break ties by $\sigma$)
            \STATE Find an alternating path $p=(g_0, i_1, g_1,\dots,i_s,g_s)$
                in $G_{pref}$
                such that $g_s$ is unallocated and
                $\psi(i_0,g_0)$ is smallest possible
            \STATE $A \leftarrow$ $A$ updated along $p$
            \STATE $A_{i_0} \leftarrow A_{i_0} \cup \{g_0\}$
            \STATE $U \leftarrow \{ i \in U : \textsc{CheckCriteria}(i, A, U, F) = \mbox{\textit{True}}\}$ (\cref{alg:checkcriteria:goods:efx})
        \ENDWHILE
    \RETURN $A$
    \end{algorithmic}
\end{algorithm}

\begin{algorithm}[t]\small
        \caption{Checking 
        fairness criteria}
    \label{alg:checkcriteria:goods:efx}
    \begin{algorithmic}[1]
        \REQUIRE An agent $i$, a partial alloc. $A$, a set $U$, and criteria $F$
        \STATE \textbf{if} $F=null$ \textbf{then} \textbf{return} \textit{True}
        \IF{$F=\EFX{}$ \textbf{or} $F=\EFX{}+\MMS{}$}
            \STATE $U_e \leftarrow$ a 
            strongly connected component 
            with no incoming edge in $G_{envy}(A, U)$
            (if there is more than one, take the one with the earliest agent according to $\sigma$)
            \STATE \textbf{if} $i \in U_e$ \textbf{then} \textbf{return} \textit{True}
        \ENDIF
        \IF{$F = \MMS{}$ \textbf{or} $F=\EFX{}+\MMS{}$}
            \STATE \textbf{if} $s_i(A_i) \le_{lex} \MMS{}_i$ \textbf{then} \textbf{return} \textit{True}
        \ENDIF  
        \RETURN \textit{False}
    \end{algorithmic}
\end{algorithm} 

\paragraph{Description.}
Fix an ordering of agents, $\sigma$. 

We initialize the algorithm with empty bundles. 
In $U$ we keep the set of prioritized agents,
for which we put $U = N$ at the start. 
In each iteration of the main loop (lines 2--8),
we take an agent $i_0$ from $U$
that has the least number of goods
(the earliest in ordering $\sigma$
if there is a tie).
In case $U=\emptyset$
(which is only possible if $F=\MMS{}$),
we take $i_0$ from $N$.
Next, agent $i_0$ chooses
the most preferred available good, $g_0$.
We update the allocation along alternating path
from $g_0$ to an unallocated good
(if $g_0$ is allocated), 
and give $g_0$ to $i_0$.

Then, we update set $U$ using
subroutine \textsc{CheckCriteria} (\cref{alg:checkcriteria:goods:efx})
to give the priority to agents
satisfying given criteria $F$.
If $F=\EFX{}$, 
in potential envy graph $G_{envy}(A,U)$ 
we look for a strongly connected component
without an incoming edge
and remove from $U$ agents outside it.
Intuitively, since there are no edges to
the strongly connected component,
no agent outside of $U$
will ever start envying an agent in $U$
at the later steps of the algorithm.
If $F=\MMS{}$,
we remove from $U$ agents that crossed their \MMS{} thresholds.
Finally, if $F=\EFX{}+\MMS{}$, we remove from $U$ agents satisfying both conditions simultaneously, and if $F=\textit{null}$, we never remove agents from $U$.
After $m$ iterations,
all goods are allocated and we output the final allocation.

Let us illustrate the most restrictive version of \cref{alg:goods:efx}, i.e., when $F=\EFX{}+\MMS{}$, with 
an example.
We provide an analogous example for 
the other criteria in \cref{app:goods:examples}.
For each variable used in the algorithm,
we will write it with superscript $t$,
to denote its value at the end of $t$-th
iteration of the main loop of our algorithm (lines 2--8). We use superscript $0$ to denote 
the initial values of variables. 

\begin{example}
    \label{ex:goods:efx}
    Consider a goods-only instance with three agents, four goods, and the following preferences (see \cref{fig:goods:ex:g_pref} for an illustration of the respective preference graph).
    \begin{align*}
        1 &: \quad \{g_1, \underline{g_2}\} \rhd \{g_3, g_4\} \\
        2 &: \quad \underline{g_1} \rhd \{g_2, g_3, g_4\} \\
        3 &: \quad g_1 \rhd \{g_2, \underline{g_3},
        \underline{g_4}\} 
    \end{align*} 
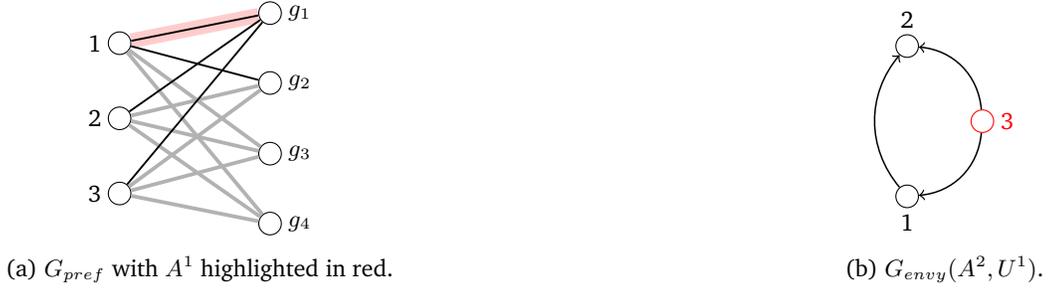
\begin{figure}[t]
    \centering
    \begin{subfigure}[t]{0.5\textwidth}
        \centering
        \begin{tikzpicture}
            \tikzset{
                vertex/.style={circle, draw=black, minimum size=0.3cm, inner sep=0},
                vertex_scc/.style={circle, draw=red, minimum size=0.3cm, inner sep=0},
                highlight/.style={draw = black!30, line width = 0.05cm,-},
                highlight_thick/.style={draw = black, line width = 0.025cm,-},
                highlight_alloc/.style={draw = red!20, line width = 0.2cm,-},
            }
            
            \node [vertex, label = {[label distance=-0.05cm]180:\footnotesize 1}] (1) at (0,2) {};
            \node [vertex, label = {[label distance=-0.05cm]180:\footnotesize 2}] (2) at (0,1) {};
            \node [vertex, label = {[label distance=-0.05cm]180:\footnotesize 3}] (3) at (0,0) {};

            \node [vertex, label = {[label distance=-0.05cm]0:\footnotesize $g_1$}] (g1) at (2,2.4) {};
            \node [vertex, label = {[label distance=-0.05cm]0:\footnotesize $g_2$}] (g2) at (2,1.47) {};
            \node [vertex, label = {[label distance=-0.05cm]0:\footnotesize $g_3$}] (g3) at (2,0.53) {};
            \node [vertex, label = {[label distance=-0.05cm]0:\footnotesize $g_4$}] (g4) at (2,-0.4) {};

            \path[draw] (1) edge[highlight_alloc](g1);
            
            \path[draw] (1) edge[highlight_thick] (g1);
            \path[draw] (1) edge[highlight_thick] (g2);
            \path[draw] (1) edge[highlight] (g3);
            \path[draw] (1) edge[highlight] (g4);
            
            \path[draw] (2) edge[highlight_thick] (g1);
            \path[draw] (2) edge[highlight] (g2);
            \path[draw] (2) edge[highlight] (g3);
            \path[draw] (2) edge[highlight] (g4);

            \path[draw] (3) edge[highlight_thick] (g1);
            \path[draw] (3) edge[highlight] (g2);
            \path[draw] (3) edge[highlight] (g3);
            \path[draw] (3) edge[highlight] (g4);

        \end{tikzpicture}
        \caption{$G_{pref}$ with $A^1$ highlighted in red.}
        \label{fig:goods:ex:g_pref}
     \end{subfigure}
     \hfill 
     \begin{subfigure}[t]{0.3\textwidth}
        \centering
        \begin{tikzpicture}
            \tikzset{
                vertex/.style={circle, draw=black, minimum size=0.3cm, inner sep=0},
                vertex_scc/.style={circle, draw=red, minimum size=0.3cm, inner sep=0},
                curve/.style={draw = black, bend right, line width = 0.02cm,->},
                curve_scc/.style={draw = red, bend right, line width = 0.02cm,->}
            }

            \node [vertex_scc, label = {[label distance=-0.05cm]0:\footnotesize \textcolor{red}{3}}] (3) at (2,1) {};
            \node [vertex, label = {[label distance=-0.05cm]90:\footnotesize 2}] (2) at (1,2) {};
            \node [vertex, label = {[label distance=-0.05cm]270:\footnotesize 1}] (1) at (1,0) {};

            \path[draw] (1) edge[curve][out=40,in=140] (2);
            \path[draw] (3) edge[curve][out=40,in=140] (1);
            \path[draw] (3) edge[curve][out=-40,in=-140] (2);
        \end{tikzpicture}      
        \caption{$G_{envy}(A^2,U^1)$.}
        \label{fig:goods:ex:g_envy}
    \end{subfigure}
    \caption{The preference graph $G_{pref}$ and potential envy graph $G_{envy}(A^2, U^1)$ for the instance from \cref{ex:goods:efx}. Thin black edges in the graph $G_{pref}$ represent agents' first indifference classes, i.e., these $e \in E$, for which $\psi(e)=1$, and gray thick edges the second ones.
    In (b), a strongly connected component without an incoming edge, $U_e^2$, is given in red,
    which is a single node. 
    }
    \label{fig:goods:ex}
\end{figure}
Let $\sigma = (1,2,3)$. 
In iteration $1$, agent $1$ 
receives $g_1$ (see the red highlight in \cref{fig:goods:ex:g_pref}). 
Observe that good $g_1$ is still available since 
there exists an alternating path $(g_1, 1, g_2)$ 
ending in an unallocated good.
Thus, in iteration $2$, agent $1$ exchanges $g_1$ 
for $g_2$ and agent $2$ receives $g_1$.
Then, in the potential envy graph 
at the end of iteration $2$, i.e., $G_{envy}(A^2, U^1)$,
agent $3$ is the only agent that forms
a strongly connected component 
with no incoming edge, $U_e^2$
(see \cref{fig:goods:ex:g_envy}). 
Moreover, since for each of agents $1$ and $2$,
its bundle is above the respective 
\MMS{} threshold, we remove them from $U$, 
thus $U^2 = \{3\}$. 
Therefore, agent $3$ receives goods $g_3$ and $g_4$.  
The final allocation is underlined.

We note that the algorithm of \citep{aziz2023possible} returns the allocation
$(\{g_2,g_4\},\{g_1\},\{g_3\})$ in the above example.
Observe that in this allocation,
agent $3$ is envious of agent $1$,
which is not the case in the
allocation returned by \cref{alg:goods:efx}.
See \cref{app:comparison} for a detailed discussion.
\end{example}

Let us now sketch the proof of correctness of our algorithm
(the full proof can be found in \cref{app:goods:efx}).
\begin{restatable}{theorem}
{thmgoodsefx}\label{thm:goods:efx}
    Given a weakly lexicographic goods-only instance $(N,M,\rhd)$, \cref{alg:goods:efx} always returns 
    \begin{itemize}
        \item[(A)] if criteria $F=\textit{null}$,  
        an \EF{1} and \PO{} allocation;
        \item[(B)] if criteria $F=\EFX{}$, 
        an \EFX{} and \PO{} allocation;
        \item[(C)] if criteria $F=\MMS{}$, 
        an \MMS{} and \PO{} allocation; and
        \item[(D)] if criteria $F=\EFX{}+\MMS{}$,
        an \EFX{}, \MMS{} and \PO{} allocation.
    \end{itemize}
\end{restatable}
\begin{proof}[Proof (Sketch)]

We begin with the series of observations that hold for all considered by us criteria $F$ (\cref{claim_body:goods:efx:tracking,claim_body:goods:efx:basic,claim_body:goods:efx:unavailable}) or only those with either \EFX{} or \MMS{}
(\cref{claim_body:goods:efx:nonempty,claim_body:goods:efx:equal,claim_body:goods:efx:mms}).
Using them, we will prove the fairness guarantees in  each of the statements (A)--(D).
Finally, \PO{} can be easily proved using \cref{thm:po:aziz}.
In this sketch we will mainly focus on showing statement (D).

We first prove several basic observations, which we group in the following two claims.
Let $X^t$ be a set of unavailable goods at the end of iteration $t$.

\begin{claim}
    \label{claim_body:goods:efx:tracking}
    For every agent $i \in N$, its indifference classes $k,k' \in [k_i]$, and iterations $t,t' \in [m]$, such that $k' < k$ and $t' < t$, it holds that
    \begin{itemize}
        \item[(a)] an unavailable good stays  unavailable,
            i.e., $X^{t'} \subseteq X^{t}$,
        \item[(b)] agents weakly prefer their goods over the available ones, i.e.,
        if $s_i(k,A^t_i) > 0$, then $\rhd_i(k') \subseteq X^{t-1}$,
        \item[(c)] score of the bundle of an agent is (weakly) monotonic, i.e.,
        $s_i(k,A^t_i) \ge s_i(k,A^{t'}_i)$
        \item[(d)] if an agent has an unavailable good, then the whole indifference class with this good is unavailable, i.e., if $\rhd_i(k, A^t_i) \cap X^t \neq \emptyset$, then $\rhd_i(k) \subseteq X^t$.
    \end{itemize}
\end{claim}

\begin{claim}
\label{claim_body:goods:efx:basic}
For arbitrary iterations $t', t \in [m]$
and agents $i,j \in N$ such that
$t' < t$,
it holds that
    \begin{itemize}
        \item[(a)] no agent becomes prioritized again, i.e., $U^{t} \subseteq U^{t'}$,
        \item[(b)] the sizes of bundles of agents in $U$ are as equal as possible, i.e., $0 \le |A^t_i| - |A^t_j| \le 1$, if $i, j \in U^t$ and $\sigma(i) < \sigma(j)$, and
        \item[(c)] an agent cannot start envying an agent that did not receive more goods, i.e., if $A^{t'}_j \succeq_j A^{t'}_i$ and $|A^{t'}_i|=|A^t_i|$, then $A^{t}_j \succeq_j A^{t}_i$.
        \item[(d)] an agent never envies agents that are later in the ordering, i.e., if $\sigma(i) < \sigma(j)$, then $A^t_i \succeq_i A^t_j$.
    \end{itemize}
\end{claim}

Next, we prove a quite technical result
that if an agent has only unavailable goods
and potentially envies other agent,
then the envied agent has at least the
same amount of unavailable goods.

\begin{claim}\vspace{-.5em}
    \label{claim_body:goods:efx:unavailable}
    For every iteration $t \in [m]$
    and agents $i,j \in N$ such that
    $A^t_i \subseteq X^t$,
    if $i$ potentially envies $j$,
    but not actually envies,
    then $|A^t_i| \le |A^t_j \cap X^t|$.
\end{claim}

Next, we focus on the case where \EFX{} is a part of criteria $F$ and show that in such a case, the set $U$ is never empty.

\begin{claim}
    \label{claim_body:goods:efx:nonempty}
    If criteria $F = \EFX{}$ or $F=\EFX{}+\MMS{}$, then for every $t \in [m]$
    it holds that $U^t \neq \emptyset$.
\end{claim}

Using this and \cref{claim_body:goods:efx:unavailable},
we show one of the key observation of the proof that allows us to guarantee \EFX{}:
if an agent, $j$, envies another agent, $i$,
with at most the same number of goods,
then there is no potential envy path from $i$ to $j$.
In consequence, in such cases, $i$ can be eliminated from $U$.

\begin{claim}
    \label{claim_body:goods:efx:equal}
    If criteria $F = \EFX{}$ or $F=\EFX{}+\MMS{}$,
    then for every agents $i,j \in N$,
    and iteration $t \in [m]$
    such that $j$ envies $i$ and $|A^t_i| \le |A^t_j|$,
    there is no path
    $p = (i_1, \dots, i_r)$ in
    $G_{envy}(A^{t-1},N)$ such that 
    $i = i_1$ and $j = i_r$.
\end{claim}

Next, we inductively show the key result 
for the proof of \MMS{}.
For this claim, for an agent $i$, let us denote set
$\rhd_i(1 \colon k)=\rhd_i(1)\cup \dots \cup \rhd_i(k)$
and vectors $s_i(1 \colon k,A^t_i)=(s_i(1,A^t_i),\dots,s_i(k,A^t_i))$
and 
$\MMS{}_i(1 \colon k)=(\MMS{}_i(1),\dots,\MMS{}_i(k))$.

\begin{claim}
    \label{claim_body:goods:efx:mms}
    If fairness criterion $F=\MMS{}$ or $F=\EFX{}+\MMS{}$, then for every iteration $t \in [m]$, agent $i \in N$, and $k \in \mathbb{N}$ such that $\rhd_i(1 \colon k) \subseteq X^t$, it holds that:
\begin{itemize}
        \item[(a)] $s_i(1 \colon k, A^t_i) \ge_{lex} \MMS{}_i(1 \colon k)$,
        \item[(b)] if $s_i(k', A^t_i) = 0$ for every $k' > k$ and
        there is an agent $j \in U^t$ such that
        $|A^t_j| < |A^t_i|$,
        then $A^t_i >_{lex} \MMS{}_i$,
        \item[(c)] if $s_i(1 \colon k, A^t_i) = \MMS{}_i(1 \colon k)$,
        then
        there is a subset of agents $L \subseteq N \setminus U^t$ such that $|L|=n - r_{k+1}$ and
        \[
        \sum_{j \in L} |A^t_j| =
            \sum_{u = 1}^k s_i(u) - r_{k+1} \cdot \sum_{u = 1}^k \MMS{}_i(u),
        \]
        where, as in \cref{prop:goods:mms:threshold},
        $r_1 = n$ and $r_{u+1} = s_i(u) - \MMS{}_i(u)\cdot r_{u}$,
        for every $u \in [k_i - 1]$.
    \end{itemize}
\end{claim}

Observe that point (a) of \cref{claim_body:goods:efx:mms} directly implies \MMS{} of the output allocation
as at the end of the algorithm all goods are unavailable.
Thus, let us show that if $F=\EFX{}+\MMS{}$, then the output allocation satisfies \EFX{} as well.
Assume otherwise, and take first $t \in [m]$ such that there exist $i,j \in N$ and $g \in A_i$ for which
$A_i \setminus \{g\} \succ_j A_j$.

From \cref{claim_body:goods:efx:basic}d
we know that $i$ is before $j$ in
ordering $\sigma$.
Let us take the last iteration $t'<t$
in which $i$ had smaller number of goods,
i.e., $|A_i^{t'}|<|A_i^t|$.
By \cref{claim_body:goods:efx:nonempty}, this means that $i \in U^{t'}$
and from \cref{claim_body:goods:efx:tracking}c
we get that $|A_i^{t'}|=|A_j^{t'}|$.
Based on this, using \cref{claim_body:goods:efx:equal},
we deduce that $j$ does not envy $i$
at the end of iteration $t'$
(otherwise it could not be that $i \in U^t$).
Since between iteration $t'$ and $t$
agent $i$ receives just one more good,
we can show that this implies
that $i$ and $j$
do not violate \EFX{}
in iteration $t$---a contradiction.
\end{proof}

\section{Chores}

In many aspects, the situation for chores is more complex than that of goods.
For example,
we can decide if there exists an \EF{} allocation for goods-only strict lexicographic instances in polynomial time,
but the same problem is NP-complete for chores~\citep{hosseini2022fairly}. 
For additive preferences, even the existence of an \EF{1} and \PO{} allocation remains an open problem for the chores case.

In this section, we will show that
\cref{alg:goods:efx}
can be adapted so that
we obtain an \EF{1} and \PO{} allocation
for every weakly 
lexicographic chores-only instance.
However, let us first discuss the
differences and challenges of finding \EFX{}
allocations for chores in comparison with goods. 
We believe that this sheds some light
on the nature of the general problem of dividing the chores.

\subsection{EFX implies MMS}
\label{subsec:chores:efxtomms}
Let us start by noting that \EFX{}
implies \MMS{} for weakly lexicographic instances, which is not the case for goods (see \cref{app:efx:mms:relation}).
To this end, we first provide a compact formulation of an agent's \MMS{} threshold analogous to the ones for goods in \cref{prop:goods:mms:threshold}.
 
\begin{proposition}
\label{prop:chores:mms:threshold}
    For every weakly lexicographic chores-only instance $(N,M,\rhd)$
    and agent $i \in N$,
    it holds that
    \[
        \MMS{}_i = ( s_i(1)/n, s_i(2)/n, \dots, \lfloor s_i(k)/ n \rfloor, 0, \dots, 0),
    \]
    where $k$ is the number of the first indifference class of $i$ such that $n$ does not divide $s_i(k)$.
    Note that since $s_i(k)$ is negative,
    the absolute value of $\lfloor s_i(k)/ n \rfloor$
    is the number of chores in the $k$-th indifference class of $i$
    divided by $n$ and \emph{rounded up}.
\end{proposition}

Using it let us show that \EFX{} implies \MMS{}.

\begin{proposition}
    \label{prop:chores:efx:implies:mms}
    Given a weakly lexicographic chores-only instance $(N,M,\rhd)$, every \EFX{} allocation satisfies \MMS{}. 
\end{proposition}

\begin{proof}
Assume for a contradiction that there is an allocation $A$ that is \EFX{} but not \MMS{}, i.e., there is agent $i \in N$ such that its $s_i(A_i) <_{lex} \MMS{}_i$. 

From \cref{prop:chores:mms:threshold}
we know that
\[
    \MMS{}_i = ( s_i(1)/n, s_i(2)/n, \dots, \lfloor s_i(k)/ n \rfloor, 0, \dots, 0),
\]
where $k$ is the number of the first indifference class of $i$ such that $n$ does not divide $s_i(k)$.

Given that for agent $i$, $s_i(A_i) <_{lex} \MMS{}_i$, it is not possible to have that $s_i(A_j) \le_{lex} \MMS{}_i$,
for every agent $j \in N \setminus \{i\}$
(otherwise, from one of the first $k$ indifference classes of agent $i$,
there would be more chores assigned to agents than there are in total).
Hence, there exists $j \in N$ such that $s_i(A_j) >_{lex} \MMS{}_i$.
Let $\bar{k}$ be an indifference class
such that $s_i(\bar{k},A_i) < \MMS{}_i(\bar{k})$,
and $c \in \rhd_i(\bar{k}, A_i)$ be 
an arbitrary chore assigned to $i$ from this class.
Observe that $s_i(A_i \setminus \{c\}) \le_{lex} \MMS{}_i <_{lex} s_i(A_j)$.
Hence, agent $i$ envies agent $j$ even after the removal of chore $c$ from its bundle.
This contradicts \EFX{}.
\end{proof}

Despite the fact that an \EFX{} allocation 
implies \MMS{}, we note that there is no implication from 
\MMS{} (with or without \PO{}) to \EFX{} 
for chores-only instances even when each agent has at most two indifference classes
(see \cref{ex:chores:mmsponotefx} in \cref{app:efx:mms:relation}).

\subsection{EFX and PO Challenges}

Given a chores-only instance, a
straightforward approach to obtain an \EFX{} and \PO{} 
allocation would be to directly copy the idea behind 
\cref{alg:goods:efx}: assign to agents their most preferred 
available items (using the alternating paths) in rounds,
removing certain agents from the set of
prioritized agents $U$ on the way,
to guarantee \EFX{}.
In the goods-only case,
the agents start picking the items from
their first indifference classes,
which works well,
as the same items are key to establishing envy relations.
However, in the chores case,
the order is reversed:
the agents start picking the items from the last indifference classes,
but still the envy relations depend predominantly
on the first indifference classes.
Hence, an initial fair allocation of chores in the last indifference classes,
may turn out to be impossible to extend to an \EFX{} allocation.

\begin{example} 
    \label{ex:chores:not_easy1}
    Consider a chores-only instance with two agents, three 
    chores and preferences as follows: 

    \begin{align*}
        1 &: \quad c_1 \rhd \{c_2, \underline{c_3}\} \\
        2 &: \quad c_1 \rhd \{\underline{c_2}, c_3\}
    \end{align*}
    If we allow agents to choose their chores,
    agent 1 can pick chore $c_3$ in iteration $1$
    and agent 2 chore $c_2$ in iteration $2$.
    Observe that the resulting partial allocation is \EFX{} (indeed, \EF{}).
    However, if we assign $c_1$ to either of the agents,
    then the recipient would envy the other agent,
    and the envy will not be eliminated by the removal of the other chore.
    Hence, this allocation cannot be extended to an \EFX{} allocation.
\end{example}

Given the above observation,
another strategy would be to start assigning chores to the agents
from their first indifference classes (worst chores).
On the other hand, for the sake of \PO{},
we need to allow agents to exchange the chores among themselves if it is beneficial for all sides.
However, such exchanges can lead to a situation similar to that in \cref{ex:chores:not_easy1}.

\begin{example} 
    \label{ex:chores:not_easy2}
    Consider a chores-only instance with four agents, five 
    chores and preferences as follows:
    \begin{align*}
        1 &: \quad \{c_1, c_2\} \rhd c_3 \rhd \{\underline{c_4}, c_5\} \\
        2 &: \quad \{c_1, c_2\} \rhd c_3 \rhd \{c_4, \underline{c_5}\} \\
        3 &: \quad \{c_4, c_5\} \rhd c_3 \rhd \{\underline{c_1}, c_2\} \\
        4 &: \quad \{c_4, c_5\} \rhd c_3 \rhd \{c_1, \underline{c_2}\}
    \end{align*} 
    Assume we start allocating chores from the first indifference classes of the agents: agent $1$ receives chore $c_1$,
    agent $2$ chore $c_2$, agent $3$ $c_4$ and $4$ $c_5$.

    Such partial allocation is not \PO{}, thus agents have to exchange their chores.
    
    As a result, each agent has one chore
    from its last indifference class.
    Such partial allocation is \EFX{} and \PO{},
    and in fact it is the only \EFX{} and \PO{} allocation of this subset of chores (up to relabeling of the agents).
    However, assigning chore $c_3$ to any one agent will result in an allocation that violates \EFX{}. 

\end{example}

Finally, we note that we can overcome the above problems and
find an \EFX{} and \PO{} allocation
when we deal with two agents only.
Indeed, \citet{gafni2021unified} have shown that for each chores-only instance we can construct an equivalent goods-only instance with $n-1$ goods for each chore of the original instance 
(the goods can be interpreted as ``not doing a chore'').
Then, each \EFX{} and \PO{} allocation in the chores-only instance correspond to an \EFX{} and \PO{} allocation in the goods-only instance, and vice-versa, assuming that no agent receives two goods that come from the same chore.
For two agents, this condition is always satisfied, thus whenever we have an algorithm guaranteeing \EFX{} and \PO{} for goods (and here we have \cref{alg:goods:efx}), we have one for chores as well.

\begin{remark}
    Given a weakly lexicographic chores-only instance $(N,M,\rhd)$ with two agents,
    we can always find an \EFX{} and \PO{} allocation in polynomial time.
\end{remark}

\subsection{EF1 and PO}
\label{subsec:chores:ef1}

In this section, we show that despite the challenges in obtaining \EFX{} and \PO{} shown in \cref{ex:chores:not_easy1,ex:chores:not_easy2},
if we relax \EFX{} to \EF{1},
then we can always find an \EF{1} and \PO{} allocation.
 
To this end,
we need to adapt the definitions from \cref{sec:goods} to chores-only setting.
Most of them are symmetrically 
applicable. To point out 
the only difference, we explicitly define alternating paths 
for chores.
Given a preference graph $G_{pref} = (N,M,E,\psi)$
and an (possibly partial) allocation $A$,
an \emph{alternating path}, $p = (c_0, i_1, c_1, \dots, i_s, c_s)$,
is a path in $G_{pref}$ 
(possibly of length zero) such that
\begin{itemize}
    \item the chores in $p$ are pairwise distinct,
    \item $p$ alternates between edges
    that belong and do not belong to $A$,
    i.e., $(c_{r-1}, i_{r}) \in A$, for every $r \in [s]$, and
    $(i_{r}, c_r) \not \in A$, for every $r \in [s]$, and
    \item for every $r \in [s]$,
    agent $i_{r}$ weakly prefers chore $c_{r}$ over $c_{r-1}$,
    i.e., $\psi(c_{r-1}, i_r) \le \psi(i_{r}, c_{r})$. 
    Note that since items are chores, an agent prefers its 
    higher indifference class, 
    which is the difference with the goods. 
\end{itemize}

\begin{algorithm}[t]\small
	\caption{Finding an EF1 and PO allocation of chores}
    \label{alg:chores:ef1}
    \begin{algorithmic}[1]
        \REQUIRE A weakly lexicographic chores-only  instance $\langle N, M, \rhd \rangle$ and the ordering $\sigma$ \\
        \ENSURE An EF1 and PO allocation $A$
        \STATE $A \leftarrow (\emptyset,\dots,\emptyset), \quad
            G_{pref} \leftarrow (N, M, E, \psi)$
        \WHILE{there is an unallocated chore}
            \STATE Take $i_0 \in N$ s.t. $|A_{i_0}|$ is minimal
            (break ties by $\sigma$)
            \STATE Find an alternating path $p=(c_0, i_1, c_1,\dots,i_s,c_s)$
                in $G_{pref}$
                such that $c_s$ is unallocated and
                $\psi(i_0,c_0)$ is largest possible
            \STATE $A \leftarrow$ $A$ updated along $p$,
            \STATE $A_{i_0} \leftarrow A_{i_0} \cup \{c_0\}$
        \ENDWHILE
    \RETURN $A$
    \end{algorithmic}
\end{algorithm} 

\paragraph{The Algorithm.}
\cref{alg:chores:ef1} proceeds similarly to \cref{alg:goods:efx} with criteria $F=$ \textit{null}.
Fix an ordering of agents, $\sigma$. We initialize the algorithm
with the empty allocation. Then, in each iteration of 
the main loop (lines 2--7), we take an agent $i$ that has 
the least number of chores (the earliest in $\sigma$ if there is a ties). Next, agent $i$ chooses the 
most preferred (i.e., least important) 
available chore, $c_0$. We update the allocation 
along alternating path from $c_0$ to an unallocated chore 
(if $c_0$ is allocated), and assign $c_0$ to $i$. 
After $m$ iterations of the main loop, all chores 
are allocated and we return the final allocation. 

We illustrate \cref{alg:chores:ef1} in the following example. 

\begin{example} 
    \label{ex:chores:ef1}
    Consider a chores-only instance with three agents, four 
    chores and preferences as follows:
    \begin{align*}
        1 &: \quad \underline{c_1} \rhd \{\underline{c_2}, c_3, c_4\} \\
        2 &: \quad \{c_1, c_2, c_3\} \rhd \underline{c_4} \\
        3 &: \quad \{c_1, c_2\} \rhd \{\underline{c_3}, c_4\} 
    \end{align*}
    Let $\sigma = (1,2,3)$. 
    In iteration $1$, agent $1$ receives $c_4$.
    However, $c_4$ is still available, so
    in iteration $2$, agent $2$ gets $c_4$, while agent $1$ exchanges $c_4$ for $c_3$ in an alternating path.

    Next, 
    agent $3$ receives $c_3$ and we update the allocation 
    along the alternating path, $p = (c_3, 1, c_2)$. Finally, 
    agent $1$ gets  $c_1$. 
\end{example}

The proof of correctness of \cref{alg:chores:ef1} is similar to the proof of \cref{thm:goods:efx}A and is relegated to \cref{app:chores:ef1}.

\begin{restatable}{theorem}{thmchoresefone}\label{thm:chores:ef1}
    Given a weakly lexicographic chores-only instance $(N,M,\rhd)$, \cref{alg:chores:ef1} always returns an \EF{1} and \PO{} allocation.
\end{restatable}

\section{Concluding Remarks}

We showed that when dealing with goods allowing agents to express indifferences immediately results in computational hardness for deciding the existence of an \EF{} allocation. Yet, we developed an algorithm that always finds an \EFX{}, \MMS{}, and \PO{} allocation in polynomial time; and more importantly, our algorithm is versatile as it enables a social planner to select a set of desired fairness criteria.
An intriguing future direction is investigating whether our positive results can be extended beyond weak lexicographic orderings, possibly to allow for more complex (e.g. combinatorial) preferences.

We illustrated the challenges in dealing with chores when agents have weakly lexicographic preferences. While in this setting \EFX{} implies \MMS{} for chores---in contrast to the goods-only case---finding an algorithm for achieving \EFX{} allocations (with or without \PO{}) remains an interesting open problem.
Nonetheless, we developed an algorithm for finding an \EF{1} and \PO{} allocation in this domain.

\section*{Acknowledgments}
Hadi Hosseini acknowledges support from NSF IIS grants \#2144413 (CAREER) and \#2107173.
We thank the anonymous reviewers for their helpful comments.

\bibliographystyle{plainnat}

\appplaceholder

\clearpage
    
\appendix

\section*{Appendix}

\section{Envy-Freeness for Goods}
\label{app:ef:np-hardness}

\thmgoodsefhardness*

\begin{proof}
    Observe that given an allocation, we can check if it is \EF{} in polynomial time. Thus, our problem is in NP. Hence, let us focus on showing the hardness of our problem. 
    
    We prove the hardness by a reduction from \textsc{Equitable Coloring} (\textsc{EC}). The \textsc{EC} problem is defined as follows. For a fixed constant $l \in \mathbb{N}$, given a graph $G$, the task is to determine if there exists a proper $l$-coloring, i.e., an assignment of $l$ colors to the vertices of the graph such that no two adjacent vertices are of the same color, with the property that the numbers of vertices in any two color classes differ by at most one. We further extend the \textsc{EC} problem by requiring that each color class has the same number of vertices. This problem is known to be NP-complete by a reduction from \textsc{Graph $k$-Colorability} \citepapp{gary1979computers,hosseini2020HEF}. 
    Let $G = (V,E)$, where $|V| = n$ and $|E| = m$.\footnote{ To maintain consistency with the standard graph notation, it should be noted that in this context, the variables $n$ and $m$ do not represent the number of agents and the number of goods as in the rest of the paper.} 
    We assume that $l \ge 3$ and that in the graph $G$ each vertex has a degree of at least two.
    Since \textsc{EC} was shown to be NP-hard for much more restricted classes of graphs~\citepapp{furmanczyk2013equitable},
    we still retain NP-hardness in this more general version of the problem.
    The assumption that each vertex has at least two incident edges, implies that there are at least that many edges as vertices, i.e.,
    \begin{equation}
    \label{condition:EFedge}
        m \geq n. 
    \end{equation}
    Finally, since for small number of vertices and fixed $l$ the problem can be solved in polynomial time by a brute-force algorithm, without loss of generality, we assume that the number of vertices is at least three times greater than the number of colors, i.e., $n \geq 3l$.
    
    For an \textsc{EC} instance, a graph $G=(V,E)$, we construct its corresponding weakly lexicographic instance $\langle N, M, \rhd \rangle$.
    To this end, we take $m+l$ agents,
    i.e., one agent, $a_e$, for each edge $e \in E$
    and one agent, $a_r$, for each color $r \in [l]$.
    Furthermore, we take $2m + n$ goods,
    where we have two goods, $e$ and $e'$, for every edge $e \in E$
    and one good, $v$, for every vertex $v \in V$.
    Thus, the set of goods is given by $M = E \cup E' \cup V$,
    where $E' = \{ e' : e \in E\}$. 
    Finally, the preferences are such that each agent has two indifference classes.
    More precisely, each edge agent $a_e$ has all the edge goods and vertex goods which are incident to the edge $e$ in the first indifference class and the rest of the goods in the second one. Meanwhile, each color agent $a_r$ has all the vertex goods in the first indifference class, while all the edge goods are in the second class.
    Formally,
    \begin{align*}
    \label{condition:EFpref}
        a_e &: \left( E \cup E' \cup \{v, v'\} \right) \rhd 
            \left( V \setminus \{v,v'\} \right),
            \quad  \mbox{for every } e=(v,v') \in E, \mbox{ and}\\
        a_r &: V \rhd \left( E \cup E' \right),
            \quad \mbox{for every } r \in [l].
    \end{align*}

    In what follows, we will show that there exists an equitable $l$-coloring of $G$ if and only if there is an EF allocation $A$ in the corresponding allocation instance. 

    First, assume that the graph $G$ admits an equitable $l$-coloring. Consider an allocation $A$ such that for every $e \in E$, edge agent $a_e$ receives two edge goods $e$ and $e'$, and for every $r \in [l]$, color agent $a_r$ receives all vertex goods of color $r$. Since all the edge agents receive exactly two goods which are from their first indifference classes, there is no envy among them. Moreover, for every $e = (v,v')$, the edge agent $a_e$ envies none of the color agents as due to the proper coloring of $G$, the vertex goods $v$ and $v'$, which are in the first indifference class of the agent $a_e$, are not allocated to the same color agent. Each color class is of the same size and therefore, each color agent receives the same number of goods. Thus, there is no envy among the color agents. Each color agent envies none of the edge agents given that each color agent gets at least $3$ vertex goods (by the assumption that $n \geq 3l$) and has all the edge goods in its second indifference set. Hence, the allocation $A$ is EF. 

    Conversely, assume that there is an EF allocation denoted by $A$. We will show that this implies the existence of a proper $l$-coloring for the graph $G$. 
    In order to establish this, we will prove a series of 
    Claims~\ref{claim1:EF} -- \ref{claim9:EF} that pertain to the properties of the allocation $A$. 

    \begin{claim_app_EF}
    \label{claim1:EF}
        For each $e \in E$, edge agent $a_e$ receives at most two edge goods under $A$.  
    \end{claim_app_EF}
    \begin{proof}
        Suppose, for a contradiction, that there is some edge agent $a_e$ who receives three or more edge goods. For every edge agent, all the edge goods are in its first indifference class. Thus, all the edge agents should receive at least three or more goods that are in their first indifference classes (as otherwise they will envy $a_e$). Therefore, there should be at least $3m$ goods to be allocated among the edge agents. By Equation~\ref{condition:EFedge}, the number of goods in the instance, $2m+n$, is bounded by $3m$. This implies that after allocating at least $3m$ goods to the edge agents, no good will be allocated to the color agents. Thus, the color agents will surely be envious---a contradiction. 
    \end{proof}

    \begin{claim_app_EF}
    \label{claim2:EF}
        For every $r \in [l]$, color agent $a_r$ receives at least one vertex good under $A$. 
    \end{claim_app_EF}
    \begin{proof}
        Suppose, for a contradiction, that there is some color agent $a_r$ who does not receive any vertex good. 
        Since its first indifference class consists exactly of all vertex goods,
        it envies an agent that receives some vertex good---a contradiction.
    \end{proof}

    \begin{claim_app_EF}
     \label{claim3:EF}
        For every $r \in [l]$, color agent $a_r$ receives at most one edge good under $A$.
    \end{claim_app_EF}
    \begin{proof}
        Suppose, for a contradiction, that there is some color agent $a_r$ who receives two or more edge goods. By Claim~\ref{claim2:EF}, this agent also receives some vertex good $v_0$. By the assumption that the graph $G$ has minimum degree two, there exists some edge $e_0=(v_0,v_1)$ incident to the vertex $v_0$ and also
        a cycle that contains $e_0$, i.e., $(e_0, e_1, \dots, e_z)$, in which $e_k = (v_k, v_{k+1})$ and $v_{z+1} = v_0$.
        Since the agent $a_r$ receives three goods from the first indifference class of $a_{e_0}$, the latter agent has to receive three such goods as well.
        Thus, by Claim~\ref{claim1:EF}, two edge goods and the vertex good $v_1$ must be assigned to the agent $a_{e_0}$.
        Similarly, the agent $a_{e_1}$ views the bundle of $a_{e_0}$ as having the score vector of at least $(3,0)$. Hence, it has to receive two edge goods and the vertex good $v_2$. 
        Continuing this way, we can show by induction that agent $a_{e_k}$ for every $k \in [z]$ is given two edge goods and vertex good $v_{k+1}$.
        However, since $v_{z+1} = v_0$ this leads to a contradiction.
    \end{proof}

    \begin{claim_app_EF}
    \label{claim4:EF}
        There exists some $e \in E$ such that edge agent $a_e$ receives two edge goods under $A$. 
    \end{claim_app_EF}
    \begin{proof}
        Observe that if none of the color agents receives an edge good, then by 
        a pigeonhole principle some edge agent has to receive at least two edge goods.
        Therefore, let us assume that there is some color agent $a_r$ who receives an edge good (by Claim~\ref{claim3:EF}, it can receive at most one edge good). 
        
        Suppose, for a contradiction, that none of the edge agents receive more than one edge good. By Claim~\ref{claim2:EF}, gent $a_r$ also receives some vertex good $v_0$. Now, consider a similar construction as in the proof of Claim~\ref{claim3:EF}. Take a cycle $(e_0, e_1, \dots, e_z)$ such that $e_k = (v_k, v_{k+1})$ and $v_{z+1} = v_0$. Since the agent $a_{e_0}$ views the bundle of $a_r$ as having the score vector of at least $(2,0)$, it must receive two goods from its first indifference class. Given that the agent $a_{e_0}$ does not receive two edge goods (by the assumption), it has to receive the vertex good $v_1$ together with one edge good. Continuing this way, we can show by induction that agent $a_{e_k}$ for every $k \in [z]$ is given one edge good and vertex good $v_{k+1}$.
        However, since $v_{z+1} = v_0$ this leads to a contradiction.
    \end{proof}

    \begin{claim_app_EF}
    \label{claim5:EF}
        For every $e \in E$, if edge agent $a_e$ has two edge goods, then it does not receive a vertex good under $A$.
    \end{claim_app_EF}
    \begin{proof}        
        Suppose, for a contradiction, that there is some edge agent $a_e$ who receives two edge goods and a vertex good. Then, every other edge agent views the bundle of the edge agent $a_e$ as having the score vector of at least $(2,1)$. Therefore, to guarantee envy-freeness of $A$, every other edge agent should receive at least three goods. Thus, at least $3m$ goods should be allocated among the edge agents. Then, there will be at most one good to be allocated to the color agents, which contradicts Claim~\ref{claim2:EF}. 
    \end{proof}

    \begin{claim_app_EF}
    \label{claim6:EF}
        For each $e \in E$, edge agent $a_e$ receives exactly two goods under $A$.
    \end{claim_app_EF}
    \begin{proof}
        Observe that if none of the color agents receives an edge good, then each edge agent has to receive two edge goods (and nothing more) by Claims~\ref{claim1:EF} and \ref{claim5:EF}. Therefore, let us assume that there is some color agent $a_r$ who receives an edge good.
        
        By Claim~\ref{claim4:EF}, there is at least one edge agent who receives two edge goods. Therefore, all edge agents should receive at least two goods. Suppose now, for a contradiction, that there is some edge agent $a_{\bar{e}}$ who receives more than two goods. By Claim~\ref{claim1:EF}, the edge agent $a_{\bar{e}}$ must receive at least one vertex good. Moreover, by Claim~\ref{claim5:EF}, this agent can not get two edge goods. Therefore, $a_{\bar{e}}$ receives at least two vertex goods. 

        Since all vertex goods are in the first indifference set of the agent $a_r$, it should receive at least two vertex goods as well. Let us denote one of these vertex goods by $v_0$. Once again, consider a similar construction to those in the proofs of Claims~\ref{claim3:EF} and \ref{claim4:EF}. Take a cycle $(e_0, e_1, \dots, e_z)$ such that $e_k = (v_k, v_{k+1})$ and $v_{z+1} = v_0$. Since the agent $a_{e_0}$ views the bundle of $a_r$ as having the score vector of at least $(2,1)$, it must receive two goods from its first indifference class and one good from its second class. By Claim~\ref{claim5:EF}, the agent $a_{e_0}$ has to receive one edge good and two vertex goods one of which is the vertex good $v_1$. Continuing this way, we can show by induction that agent $a_{e_k}$ for every $k \in [z]$ is given one edge good and two vertex goods one of which is vertex good $v_{k+1}$.
        However, since $v_{z+1} = v_0$ this leads to a contradiction.
        
    \end{proof}

    \begin{claim_app_EF}
    \label{claim7:EF}
        There exists some $r \in [l]$, the color agent $a_r$ receives at least $\frac{n}{l}$ goods under $A$. 
    \end{claim_app_EF}
    \begin{proof}
        The total number of goods in the instance is $2m+n$. By Claim~\ref{claim6:EF}, edge agents receive $2m$ goods in total. Therefore, one of the color agents must receive at least $\frac{n}{l}$ goods from the pigeonhole principle. 
    \end{proof}

    \begin{claim_app_EF}
    \label{claim8:EF}
        For each $e \in E$, edge agent $a_e$ receives exactly two edge goods and no vertex good under $A$. 
    \end{claim_app_EF}
    \begin{proof}
        Observe that if none of the color agents receives an edge good, then each edge agent should receive exactly two edge goods and nothing more by Claims~\ref{claim1:EF} and \ref{claim5:EF}. Therefore, assume now that there is some color agent $a_r$ who receives at least one edge good. 

        By Claim~\ref{claim2:EF}, the agent $a_r$ also receives vertex good $v$. By the assumption that the graph $G$ has minimum degree two, there is an edge $e = (v,v')$ which is incident to the vertex $v$. The agent $a_e$ views the bundle of the agent $a_r$ as having the score vector $(2,0)$ or better. Since the edge agent $a_e$ receives exactly two goods (by Claim~\ref{claim6:EF}), the color agent $a_r$ cannot receive any other good (otherwise, the agent $a_e$ will envy the agent $a_r$). Thus, the color agent $a_r$ also receives exactly two goods -- one edge good and the vertex good $v$. By Claim~\ref{claim7:EF}, there is some color agent $a_{r'}$ who receives at least $\frac{n}{l}$ goods, which is at least $3$ goods by the assumption that $n \geq 3l$. Since the agent $a_{r'}$ also receives one or more vertex goods (by Claim~\ref{claim2:EF}), the agent $a_r$ views the bundle of the agent $a_{r'}$ as having the score vector $(1,2)$ or better. Thus, the agent $a_r$ will be envious of the agent $a_{r'}$---a contradiction. 
    \end{proof}

    \begin{claim_app_EF}
    \label{claim85:EF}
        For each $r \in [l]$, color agent $a_r$ receives exactly $\frac{n}{l}$ vertex goods and no edge good under $A$.
    \end{claim_app_EF}
    \begin{proof}
        From Claim~\ref{claim8:EF} all edge goods are distributed among edge agents,
        and all vertex goods are distributed among color agents.
        Since there is no envy among color agents and they have identical preferences,
        every color agent must receive exactly the same number of vertex goods, i.e., $\frac{n}{l}$.
    \end{proof}
    
    \begin{claim_app_EF}
    \label{claim9:EF}
        For every $e \in E$ with $e = (v,v')$, there is no $r \in[l]$ such that the color agent $a_r$ receives both the vertex goods $v$ and $v'$ under $A$. 
    \end{claim_app_EF}
    \begin{proof}
        Suppose, for a contradiction, that there is some color agent $a_r$ who receives both of the vertex goods $v$ and $v'$ for some edge $e=(v, v')$.
        By Claim~\ref{claim8:EF},
        edge agent $a_e$ has two edge goods and nothing more.
        However, by Claim~\ref{claim85:EF}, agent $r$ has $\frac{n}{l}$ goods,
        which by our initial assumptions that $n \geq 3l$ is at least 3.
        Since both $v$ and $v'$ are in the first indifference class of $e$
        that makes agent $a_e$ envy agent $a_r$, which is a contradiction.
    \end{proof}

    To conclude the proof,
    observe that we can identify the color of each vertex with the color agent that received the corresponding vertex good.
    By Claim~\ref{claim9:EF} such coloring will be proper and by Claim~\ref{claim85:EF} it will be equitable as well. 
\end{proof}

Observe that the \EF{} allocation constructed in the proof for the NP-completeness of the \EF{}-existence problem satisfies \PO{} as well. This is because every good is assigned to an agent who has it in its first indifference set. As a result, the problem of determining whether an \EF{} and \PO{} allocation is also NP-complete.  

\begin{corollary}
    Deciding whether an \EF{} and \PO{} allocation exists for a given weakly lexicographic goods-only instance is NP-complete.
\end{corollary}

\section[Correctness of Algorithm 1]{Correctness of Algorithm~\ref{alg:goods:efx}}
\label{app:goods:efx}

\thmgoodsefx*
\begin{proof}
Throughout the proof,
we will use $t$ to denote the number of an
iteration of the main loop (lines 2--8),
and each variable used in the algorithm
with superscript $t$, e.g., $A^t$ or $U^t$,
will denote the value of this variable at the end of iteration $t$.
For convenience, by $A^0$ or $ U^0$, etc.,
we denote the initial values of 
these variables.

We start the proof by showing that \cref{alg:goods:efx} is well defined
and always returns a complete allocation for all statements (A), (B), (C) and (D) of \cref{thm:goods:efx}.
To this end, observe that in each iteration of the main loop (lines 2--8),
we take one still unallocated good, $g_s$, and assign it to 
one of the agents identified by \cref{alg:checkcriteria:goods:efx}
(possibly, changing the assignments of some of the
already allocated goods on the way).
Hence, after $m$ iterations of the main loop,
we allocate all items.

The remainder of the proof is organized 
into the three main parts as follows: 
\begin{enumerate}
    \item First, we will prove several basic
    properties regarding the concepts and
    notions used in the proof, 
    which hold for all considered criteria $F$ 
    (\cref{claim:goods:efx:tracking,claim:goods:efx:basic,claim:goods:efx:unavailable}).
    \item Next we will show several results that make some assumptions on the criteria we consider (\cref{claim:goods:efx:nonempty,claim:goods:efx:equal,claim:goods:efx:mms}).
    \item Then, for each statement, 
    we will provide a separate 
    proof for the fairness of 
    the output allocation: 
    \EF{1} for (A),
    \EFX{} for (B), 
    \MMS{} for (C), and
    \EFX{} and \MMS{} for (D). 
    \item Finally, we will prove 
    \PO{} of the output allocation 
    once for all statements. 
\end{enumerate}

\noindent
\textbf{\underline{Part 1: Basic common properties.}}

\noindent
By $X^t$ let us denote the set of
all \emph{unavailable} goods
in iteration $t$, i.e.,
goods from which there is
no alternating path to an unallocated good
in $G_{pref}$.
Let us prove the basic properties regarding unavailable goods and score vectors of agents.


\begin{claim_app_EFX}
    \label{claim:goods:efx:tracking}
    For every agent $i \in N$, its indifference classes $k,k' \in [k_i]$, and iterations $t,t' \in [m]$, such that $k' < k$ and $t' < t$, it holds that
    \begin{itemize}
        \item[(a)] an unavailable good stays  unavailable,
            i.e., $X^{t'} \subseteq X^{t}$,
        \item[(b)] agents weakly prefer their goods over the available ones, i.e.,
        if $s_i(k,A^t_i) > 0$, then $\rhd_i(k') \subseteq X^{t-1}$,
        \item[(c)] the score of the own bundle of an agent is (weakly) monotonic, i.e.,
        $s_i(k,A^t_i) \ge s_i(k,A^{t'}_i)$,
        \item[(d)] if an agent has an unavailable good, then the whole indifference class with this good is unavailable, i.e., if $\rhd_i(k, A^t_i) \cap X^t \neq \emptyset$, then $\rhd_i(k) \subseteq X^t$.
    \end{itemize}

\end{claim_app_EFX}
\begin{proof}
    For \textit{(a)},
    we will show
    that for every good $g \in M$,
    if $g \in X^t$,
    then $g \in X^{t+1}$,
    which will imply the thesis by induction.
    If $g \in X^t$, then,
    for allocation $A^{t}$,
    there is no alternating path in $G_{pref}$ that
    starts in $g$ and ends in an unallocated good.
    Hence, the same is true for every good $g'$,
    for which there exists an alternating path that
    starts in $g$ and ends in $g'$. 
    This means that in the next iteration,
    the allocation of every such good $g'$ as well as good $g$
    has to remain the same
    (only the available goods can change ownership in lines 5 and 6 of the algorithm).
    Thus, again, for allocation $A^{t+1}$,
    there is no alternating path in $G_{pref}$
    starting in $g$ and ending in an unallocated good.
    Thus, $g \in X^{t+1}$.

    For \textit{(b)},
    assume by contradiction that
    there exists an available good
    that is strictly preferred by agent $i$
    over one of its owned goods, i.e.,
    there exist $g \in M \setminus X^{t-1}$ and $g' \in A^t_i$ such that
    $g \succ_i g'$.
    Let $k$ be an indifference class of good $g'$, i.e., $g' \in \rhd_i(k)$, and $t'$ be the first iteration in which $s_i(k',A^{t'}_i) > 0$ for some $k' \ge k$.
    Since by an exchange in an alternating path each agent always weakly prefers its new goods over the old ones,
    this means that in iteration $t'$    
    agent $i$ was chosen in line 3 of the algorithm
    and received a new good,
    $g'' \in \rhd_i(k')$,
    in line 5.
    Observe also that $t' \le t$,
    which means, by \textit{(a)}, that $g$ was available at the end of iteration $t'-1$.
    However, since $g \succ_i g' \succeq_i g''$, agent $i$ should receive good $g$ instead of $g''$ in iteration $t'$---a contradiction.

    For \textit{(c)},
    observe that the bundle of an agent can change only in line 5 and 6 of the algorithm.
    In line 5 an agent receives an additional good, thus the score $s_i(k,A^t_i)$ can only increase.
    Hence, let us prove that it is not possible for a score $s_i(k,A^t_i)$ to decrease as an effect of the exchange along the alternating path in line 6 of the algorithm.
    To the end, we will show that in the alternating paths an agent exchanges only goods it sees as indifferent.
    More formally, let $p = (g_0, i_1, g_1, \dots, i_s, g_s)$ be an alternating path chosen in line 4 of the algorithm.
    We will prove that for every $j \in [s]$,
    it holds that $g_{j-1} \sim_{i_j} g_j$.
    
    Observe that $g_{j} \succeq_{i_j} g_{j-1}$ by the the definition of an alternating path.
    Observe also that $g_{j}$ is available at the end of iteration $t-1$, as witnessed by alternating path $p'=(g_{j}, i_{j+1}, g_{j+1},\dots,i_s,g_s)$.
    Thus, from \textit{(b)} we get that $g_{j-1} \succeq_{i_j} g_{j}$, which proves that
    $g_{j-1} \sim_{i_j} g_j$.

    For \textit{(d)},
    assume for a contradiction that
    there exists $g \in \rhd_i(k,A^t_i) \cap X^t$ and $g' \in \rhd_i(k) \setminus X^{t}$.
    This means that there exists an alternating path
    $p = (g_0,i_1,g_1,\dots,i_r,g_r)$ in $G_{pref}$
    that starts in $g_0 = g'$ and
    ends in some unallocated good $g_r$.
    Thus, since $i$ is indifferent between $g$ and $g'$,
    we get that path
    $(g, i, g_0,i_1,g_1,\dots,i_r,g_r)$
    is also alternating.
    As a result, $g$ is available---a contradiction.
\end{proof}


Next, let us prove some basic observations regarding set $U$ and envy in our algorithm.
\begin{claim_app_EFX}
\label{claim:goods:efx:basic}
For arbitrary iterations $t', t \in [m]$
and agents $i,j \in N$ such that
$t' < t$,
it holds that
    \begin{itemize}
        \item[(a)] no agent becomes prioritized again, i.e., $U^{t} \subseteq U^{t'}$,
        \item[(b)] the sizes of bundles of agents in $U$ are as equal as possible, i.e., $0 \le |A^t_i| - |A^t_j| \le 1$, if $i, j \in U^t$ and $\sigma(i) < \sigma(j)$, and
        \item[(c)] an agent cannot start envying an agent that did not receive more goods, i.e., if $A^{t'}_j \succeq_j A^{t'}_i$ and $|A^{t'}_i|=|A^t_i|$, then $A^{t}_j \succeq_j A^{t}_i$.
        \item[(d)] an agent never envies agents that are later in the ordering, i.e., if $\sigma(i) < \sigma(j)$, then $A^t_i \succeq_i A^t_j$.
    \end{itemize}

\end{claim_app_EFX}
\begin{proof}
We get~\textit{(a)} directly from the algorithm
as $U^t$ is always chosen as a subset of $U^{t-1}$ in line 7 
by checking criteria $F$ (\cref{alg:checkcriteria:goods:efx}).

For~\textit{(b)},
since in line 3 of the algorithm we choose an agent from $U$
(whenever it is nonempty) with the smallest number of goods and
from \textit{(a)} $U$ never gets new agents,
the difference in the number of goods of agents in $U$
never increases above one.
Since we break ties according to $\sigma$,
if the number of goods is not equal,
it is the earlier agent in $\sigma$ that has
one more good.

For~\textit{(c)},
assume otherwise, i.e.,
$A^{t'}_j \succeq_j A^{t'}_i$ and $|A^{t'}_i|=|A^t_i|$,
but $A^{t}_i \succ_j A^{t}_j$.
From \cref{claim:goods:efx:tracking}c we have
$A^{t}_i \succ_j A^{t}_j \succeq_j A^{t'}_j \succeq_j A^{t'}_i$.
Thus, in order for the envy to appear,
between iteration $t'$ and $t$,
agent $i$ had to exchange one of its goods for a good, $g$,
which is preferred by $j$ over one of its goods, $g'$, in $A^{t'}_j$.
This implies that $g$ was available in iteration $t'$.
However, this and the fact that $g \succ_j g'$
contradicts \cref{claim:goods:efx:tracking}b.

Finally, for (d), assume by contradiction that there is an iteration $t$ such that $A^t_j \succ_i A^t_i$ and let us take the first such $t$.
Since $A^{t-1}_i \succeq_i A^{t-1}_j$,
in iteration $t$,
agent $j$ obtains good $g$
(either receives it in line 7 or exchanges
for another good in the alternating path in line 6)
that is preferred by agent $i$ over one of $i$'s goods.
Since $j$ received $g$ in iteration $t$
it must be that $g$ was available at the end of iteration $t-1$, i.e., $g \not \in X^{t-1}$.
However, this contradicts \cref{claim:goods:efx:tracking}b.
\end{proof}


Finally, let us show a quite technical result
that if an agent has only unavailable goods
and potentially envies other agent,
then the envied agent has at least the
same amount of unavailable goods.

\begin{claim_app_EFX}
    \label{claim:goods:efx:unavailable}
    For every iteration $t \in [m]$
    and agents $i,j \in N$ such that
    $A^t_i \subseteq X^t$,
    if $i$ potentially envies $j$,
    but not actually envies,
    then $|A^t_i| \le |A^t_j \cap X^t|$.
\end{claim_app_EFX}
\begin{proof}
    Let us denote the set of available goods
    in iteration $t$ by $H^t = M \setminus X^t$.
    By definition,
    if $i$ potentially envies $j$,
    then it strictly prefers bundle
    $A^t_j \cup H^t$ over its bundle.
    This means that there exists a position $\bar{k}$
    such that
    $s_i(\bar{k},A^t_j \cup H^t) > s_i(\bar{k},A^t_i)$ and
    $s_i(k,A^t_j \cup H^t)=s_i(k,A^t_i)$, for every $k < \bar{k}$.
    By $k^*$ let us denote the last position
    for which $s_i(k^*,A^t_i) > 0$.
    From \cref{claim:goods:efx:tracking}b and \cref{claim:goods:efx:tracking}d
    we know that every good in the first $k^*$
    indifference classes of $i$
    is unavailable.
    Thus, for every $k \le k^*$, we have
    $s_i(k,A^t_j) = s_i(k,A^t_j \cup H^t)$.
    Since $i$ does not actually envy $j$,
    this means that $\bar{k} > k^*$
    (otherwise, we would have
    $s_i(\bar{k},A^t_j) = 
    s_i(\bar{k},A^t_j \cup H^t) > 
    s_i(\bar{k}, A^t_i)$)
    Hence, we get the bound in question, i.e.,
    \(
        |A^t_j \cap X^t| \ge 
        \sum_{k=1}^{k^*} s_i(k,A^t_j) = 
        \sum_{k=1}^{k^*} s_i(k,A^t_i) = 
        |A^t_i|.
    \)
\end{proof}

\noindent
\textbf{\underline{Part 2: Criteria specific properties.}}

\noindent
We now move the second part of the proof 
in which  
we consider properties that rely on specific fairness criteria.


\begin{claim_app_EFX}
    \label{claim:goods:efx:nonempty}
    If criteria $F = \EFX{}$ or $F=\EFX{}+\MMS{}$, then for every $t \in [m]$
    it holds that $U^t \neq \emptyset$.
\end{claim_app_EFX}
\begin{proof}
    Observe that in every directed graph,
    there exists a strongly connected component without any incoming edge.
    In particular, such a component exists in $G_{envy}(A^t,S)$ for every $S \subseteq N$ and $t \in [m]$.
    Hence, in each iteration $t \in [m]$,
    for some agent $i \in U^{t-1}$,
    \cref{alg:checkcriteria:goods:efx}
    will return \textit{True} and keep it in $U^t$.
\end{proof}


Now, let us show one of the key observation of the proof that allows us to guarantee \EFX{}:
if an agent, $j$, envies another agent, $i$,
with at most the same number of goods,
then there is no potential envy path from $i$ to $j$ in $G_{envy}$.
We will use this result later on to prove that, in such cases, $i$ can be eliminated from $U$.

\begin{claim_app_EFX}
    \label{claim:goods:efx:equal}
    If criteria $F = \EFX{}$ or $F=\EFX{}+\MMS{}$,
    then for every agents $i,j \in N$,
    and iteration $t \in [m]$
    such that $j$ envies $i$ and $|A^t_i| \le |A^t_j|$,
    there is no path
    $p = (i_1, \dots, i_r)$ in
    $G_{envy}(A^{t-1},N)$ such that 
    $i = i_1$ and $j = i_r$.
\end{claim_app_EFX}
\begin{proof}
    We will prove that
    for every agents $i,j \in N$,
    if $t \in [m]$
    is the first iteration in which $j$ envies $i$ and $|A^t_i| \le |A^t_j|$,
    then in iteration $t-1$, it holds that:
\begin{itemize}
        \item[(a)] agent $j$ envies $i$, i.e., $A^{t-1}_i \succ_j A^{t-1}_j$,
        \item[(b)] agent $i$ has one more good than 
    agent $j$ in its bundle, i.e., $|A^{t-1}_i| = s$ and $|A^{t-1}_j| = s-1$
        \item[(c)] all goods of agent $i$
    are unavailable, i.e., $A_i^{t-1} \subseteq X^{t-1}$, and
        \item[(d)] there is no path
    $p = (i_1, \dots, i_r)$ in
    $G_{envy}(A^{t-1},N)$ such that 
    $i = i_1$ and $j = i_r$.
    \end{itemize}
    Observe that point (d) will imply the thesis of the claim, as by \cref{claim:goods:efx:tracking}a and \cref{claim:goods:efx:tracking}c , it is not possible for the potential envy to reappear in our algorithm once it is gone.
    Thus, if there is no potential envy path from $i$ to $j$ in iteration $t$,
    then there is no such path in any later iteration.

    Assume otherwise.
    Fix arbitrary agents $i,j \in N$,
    and take the first iteration  $t \in [m]$
    such that
    $j$ envies $i$ and
    $|A^t_i| \le |A^t_j|$.
    Observe that from \cref{claim:goods:efx:basic}d we know that
    $i$ is before $j$ in the ordering $\sigma$.
    Denote $s = |A^t_i|$.
    Let us now prove each of the points (a)--(d) separately.

    \textit{Point (a).}
    Let us start by showing that actually
    $|A^t_j| = |A^t_i| = s$.
    Assume otherwise, i.e., $|A^t_j| > s$.
    Let $t'$ be the last iteration in which $|A^{t'}_j| = s$.
    This means that in iteration $t' + 1$
    agent $j$ is chosen in line 3 of the algorithm,
    not agent $i$.
    Since $|A^{t'}_j| = s \ge |A^{t'}_i|$ and $\sigma(i) < \sigma(j)$,
    this implies that $i \not \in U^{t'}$ and $j \in U^{t'}$.
    From the minimality of $t$,
    we know that $j$ does not envy $i$ in iteration $t'$, i.e., $A^{t'}_j \succeq_j A^{t'}_i$.
    Hence, $|A^{t'}_i| \neq s$, as otherwise this would violate \cref{claim:goods:efx:basic}c.
    Thus, $|A^{t'}_i| < s$.
    However, since $i \not \in U^{t'}$
    and the fact that, by \cref{claim:goods:efx:nonempty}, set $U$ is never empty, the number of goods of $i$ cannot increase in later iterations.
    But this contradicts the fact that $|A^t_i|=s$.
    
    Let us now show that agent $j$
    envies $i$ in iteration $t-1$ as well.
    Assume otherwise, i.e., the envy appeared in iteration $t$.
    From \cref{claim:goods:efx:basic}c this means that
    $i$ received a new good in line 7 of iteration $t$,
    which means that $|A^{t-1}_i| = s-1$ and $|A^{t-1}_j|=s$.
    Consequently, this implies that $i \in U^{t-1}$ 
    from the fact that $U$ is never empty (\cref{claim:goods:efx:nonempty}).
    Furthermore, from \cref{claim:goods:efx:basic}b we get that $j \not \in U^{t-1}$.
    Let $t'$ be the last iteration in which $j \in U^{t'}$.
    By \cref{claim:goods:efx:basic}b, $|A^{t'}_j| \le |A^{t'}_i| \le |A^{t-1}_i| = s-1$.
    Observe that between iterations $t'+1$ and $t$ agent $j$ cannot receive new goods in line 7 of the algorithm
    (as $i$ should receive a good before $j$).
    Hence,
    $|A^{t' + 1}_j| = |A^{t}_j|$ and the latter is equal to $s$ as we have shown in the previous paragraph.
    Thus, in iteration $t'$ we have
    $|A^{t'}_j| = |A^{t'}_i| = s-1$.
    But then, it is agent $i$ that should receive a new goods in line 7 of iteration $t'+1$ not $j$---a contradiction.

    \textit{Point (b).}
    Since, by (a), agent $j$ envies $i$ in iteration $t-1$,
    then by minimality of $t$,
    this means that $|A^{t-1}_i| > |A^{t-1}_j|$.
    Since in each single iteration,
    the number of goods increases by one for one agent
    and stays the same for all others,
    it must be that
    $|A^{t-1}_i| = s$ and $|A^{t-1}_j| = s -1$.

    \textit{Point (c).}
    Now let us show that at the end of iteration $t-1$ all goods of agent $i$ are unavailable.
    Assume for contradiction that 
    $A^{t-1}_i \not \subseteq X^{t-1}$.
    By $k^*$ let us denote the last position
    for which $s_j(k^*, A^{t-1}_j)>0$.
    Since, by (a), $j$ envies $i$
    there must be a position $\bar{k} \in \mathbb{N}$
    such that
    $s_j(\bar{k}, A^{t-1}_i) > s_j(\bar{k}, A^{t-1}_j)$
    and $s_j(k, A^{t-1}_i) = s_j(k, A^{t-1}_j)$,
    for every $k < \bar{k}$.
    Now, let us consider two cases
    based on relation between $k^*$ and $\bar{k}$.

    If $\bar{k} \ge k^*$, then
    since $j$ has only one less good than $i$,
    $\bar{k}$ is the only position on which
    score vectors $s_j(A^{t-1}_i)$ and $s_j(A^{t-1}_j)$ differ.
    Specifically,
    $s_j(\bar{k}, A^{t-1}_i) = s_j(\bar{k}, A^{t-1}_j) + 1$
    and $s_j(k, A^{t-1}_i) = s_j(k, A^{t-1}_j)$,
    for every $k < \bar{k}$.
    Take arbitrary $i$'s available good,
    $g \in A^{t-1}_i \setminus X^{t-1}$.
    From \cref{claim:goods:efx:tracking}b we know that $g$
    must be in the $\bar{k}$-th or $k^*$-th
    indifference class of agent $j$.
    In line 7 of iteration $t$ agent $j$ receives
    the most preferred available good,
    which, since $g$ is available,
    has to be in $\bar{k}$-th or earlier indifference class.
    Thus, $s_j(A^{t}_j) \ge_{lex} s_j(A^{t}_i)$,
    but this contradicts the fact that $j$ envies $i$.

    Now, consider the case in which 
    $\bar{k} < k^*$.
    By \cref{claim:goods:efx:tracking}b,
    all $i$'s goods in $\bar{k}$-th
    and earlier indifference classes
    of $j$ are unavailable.
    Thus, if we remove all available goods
    we still have
    $s_j(A^{t-1}_i \cap X^{t-1}) >_{lex} s_j(A^{t-1}_j)$.
    Now, let $t' < t-1$ be some iteration in which
    $|A^{t'}_j| = |A^{t'}_i| = |A^{t-1}_i \cap X^{t-1}|$
    (from \cref{claim:goods:efx:basic}b and \cref{claim:goods:efx:nonempty}
    we know that there exists at least one).
    By \cref{claim:goods:efx:tracking}c,
    $A^{t-1}_j \succeq_j A^{t'}_j$,
    thus
    $A^{t-1}_i \cap X^{t-1} \succ_j A^{t'}_j$.
    On the other hand,
    from the minimality of $t$
    we know that $j$ does not envy $i$
    in iteration $t'$, i.e.,
    $A^{t'}_j \succeq_j A^{t'}_i$.
    Since $j$ weakly prefers $A^{t'}_j$
    over $A^{t'}_i$ but not over $A^{t-1}_i \cap X^t$
    and we know that $|A_i^{t'}| = |A^{t-1}_i \cap X^{t-1}|$,
    there must be a good in the latter but not former,
    $g \in (A^{t-1}_i \cap X^{t-1}) \setminus A^{t'}_i$,
    that $j$ strictly prefers
    over one of its goods $g' \in A^{t'}_j$.
    However, since $g$ is assigned to agent $i$
    in some iteration later than $t'$,
    $g$ is available in $t'$.
    But then the fact that $g \succ_j g'$ 
    contradicts \cref{claim:goods:efx:tracking}b.
    Therefore, all of the goods
    in $A^{t-1}_i$ are unavailable.
    
    \textit{Point (d)}.
    In order to prove this point,
    let us take minimal iteration $t \in [m]$ for which there are agents $i,j \in N$ such that $j$ envies $i$, $|A^t_j| \ge |A^t_i|$ and there is a
    path $p = (i_1, i_2, \dots, i_r)$ in $G_{envy}(A^{t-1},N)$ such that $i = i_1$ and $j = i_r$.
    We will show by induction that each agent in $\{i_1, i_2, \dots, i_r\}$ has $|A_i^t|=s$ unavailable goods in iteration $t-1$.
    This will lead to contradiction as from (b) we know that agent $j = i_r$ has $s-1$ goods in total.
    
    The basis of induction follows from (c).
    For the inductive step, assume that $i_u$ has
    $s$ unavailable goods
    for arbitrary $u \in [r-1]$.
    Observe that since in iteration $t-1$ we have that $j \in U^{t-1}$ and it has $s-1$ goods, then by \cref{claim:goods:efx:basic}b,
    no agent has more than $s$ goods.
    Thus, all of the goods of $i_u$ are unavailable.
    Hence, by \cref{claim:goods:efx:unavailable},
    either $i_u$ actually envies $i_{u+1}$ or
    $i_{u+1}$ also has $s$ unavailable goods.
    In the former case, by minimality of $t$,
    we get that there is no potential envy path from $i_{u+1}$ to $i_u$ in 
    $G_{envy}(A^{t-2},N)$.
    Since $j$ potentially envies $i$ in $t-2$ (as otherwise it would not envy it in $t-1$), this means that there is no
    potential envy path from $i_{u+1}$ to $j$ in $G_{envy}(A^{t-1},N)$,
    which contradicts our assumptions.
    Hence, $i_{u+1}$ has $s$ unavailable goods.
    Thus, by induction, we obtain that $j$
    also has $s$ unavailable goods,
    but that contradicts the fact
    that $|A^{t-1}_j|=s-1$, which concludes the proof.
\end{proof}


Next, we show the key result 
for the proof of \MMS{}.
Specifically, we consider an arbitrary agent $i$ and iteration $t$ such that all goods in the first $k$ indifference classes of agent $i$ are unavailable.
We show that in such a case
the score vector of $i$'s bundle is lexicographically as large as its MMS vector up to first $k$ coordinates (a).
Moreover, if all of $i$'s goods are from its first $k$ classes and there is an agent with less goods than $i$, then $i$ strictly prefers its bundle to its MMS (b).
Finally, if the first $k$ positions of the score of $i$'s bundle are exactly the same as the first $k$ positions of $i$'s MMS,
then we can identify a certain subset of agents that are not in $U$ with a specific number of goods in total (c).
To this end, let us introduce some additional notation.
For an agent $i$, let us denote set
$\rhd_i(1 \colon k)=\rhd_i(1)\cup \dots \cup \rhd_i(k)$
and vectors $s_i(1 \colon k,A^t_i)=(s_i(1,A^t_i),\dots,s_i(k,A^t_i))$
and 
$\MMS{}_i(1 \colon k)=(\MMS{}_i(1),\dots,\MMS{})_i(k))$.

\begin{claim_app_EFX}
    \label{claim:goods:efx:mms}
    If fairness criterion $F=\MMS{}$ or $F=\EFX{}+\MMS{}$, then for every iteration $t \in [m]$, agent $i \in N$, and $k \in \mathbb{N}$ such that $\rhd_i(1 \colon k) \subseteq X^t$, it holds that:
\begin{itemize}
        \item[(a)] $s_i(1 \colon k, A^t_i) \ge_{lex} \MMS{}_i(1 \colon k)$,
        \item[(b)] if $s_i(k', A^t_i) = 0$ for every $k' > k$ and
        there is an agent $j \in U^t$ such that
        $|A^t_j| < |A^t_i|$,
        then $A^t_i >_{lex} \MMS{}_i$,
        \item[(c)] if $s_i(1 \colon k, A^t_i) = \MMS{}_i(1 \colon k)$ and $i \in U^t$,
        then
        there is a subset of agents $L \subseteq N \setminus U^t$ such that $|L|=n - r_{k+1}$ and
        \[
        \sum_{j \in L} |A^t_j| =
            \sum_{u = 1}^k s_i(u) - r_{k+1} \cdot \sum_{u = 1}^k \MMS{}_i(u),
        \]
        where, as in \cref{prop:goods:mms:threshold},
        $r_1 = n$ and $r_{u+1} = r_u - (s_i(u) - \MMS{}_i(u)\cdot r_{u})$,
        for every $u \in [k_i - 1]$.
    \end{itemize}
\end{claim_app_EFX}
\begin{proof}
    We will prove the claim by induction on the iteration number $t$.
    Observe that for $t=0$,
    for every agent $i$ we have $\rhd_i(1)\not \subseteq X^0$.
    Hence, the thesis holds vacuously.
    Thus, let us take $t>0$
    and assume that for every $t' \in [t-1]$
    the thesis holds.
    In what follows, we will show that each point (a)--(c) also holds in iteration $t$.

    \textit{Point (a)}.
    Assume otherwise, i.e., there is $k \in \mathbb{N}$ such that $\rhd_i(k') \subseteq X^t$ for every $k' \le k$,
    but $s_i(k', A^t_i) =_{lex} \MMS{}_i(k')$,
    for every $k' \in [k-1]$, and
    $s_i(k, A^t_i) < \MMS{}_i(k)$.
    Observe that since the point (a) was satisfied in iteration $t-1$,
    it must be that there is $k^* \in [k]$
    such that $\rhd_i(k^*) \not \subseteq X^{t-1}$.
    Since $s_i(A^t_i) <_{lex} \MMS{}_i$, we know that $i \in U^t$
    and, by Claim~\ref{claim:goods:efx:basic}a, $i \in U^{t-1}$.
    Thus, from (c) for iteration $t-1$ we get that
    there is a set $L \subseteq N \setminus U^{t-1}$
    such that $|L| = n - r_{k^*}$ and
    \begin{equation}
    \label{eq:goods:efx:mms:1}
        \sum_{j \in L} |A^t_j| =
        \sum_{u = 1}^{k^*-1} s_i(u) - r_{k^*} \cdot \sum_{u = 1}^{k^*-1} \MMS{}_i(u).
    \end{equation}
    
    Let us denote $s = |A^{t}_i|$.
    Since $i \in U^{t}$,
    by \cref{claim:goods:efx:basic}b,
    no agent has more than $s+1$ goods.
    Now, let us consider two cases, based on whether $k^*=k$ (1) or $k^* < k$ (2).

    \textit{Case (a.1)}
    Assume $k^*=k$.
    Then, since by \cref{claim:goods:efx:basic}b $s_i(k,A^t_i)=0$ for every $k' \ge k$, we have
    \begin{equation}
    \label{eq:goods:efx:mms:2}
    |A^t_i| = \sum_{u=1}^{k} s_i(u,A^t_i)  < \sum_{u = 1}^{k} \MMS{}_i(u),
    \end{equation}
    while for every $j \in N \setminus L \setminus \{i\}$,
    we have
    \begin{equation}
    \label{eq:goods:efx:mms:3}
    |A^t_j| \le  |A^t_i| + 1 \le \sum_{u = 1}^{k} \MMS{}_i(u).
    \end{equation}
    Adding inequalities \eqref{eq:goods:efx:mms:1}, \eqref{eq:goods:efx:mms:2}, and \eqref{eq:goods:efx:mms:3} sidewise together, we get
    \(
    \sum_{j \in N} |A^t_j| <
        \sum_{u = 1}^{k-1} s_i(u) - r_{k} \cdot \sum_{u = 1}^{k-1} \MMS{}_i(u)
        + r_{k} \cdot \sum_{u = 1}^{k} \MMS{}_i(u),
    \)
    which gives
    \begin{align*}
    \sum_{j \in N} |A^t_j| &<
        \sum_{u = 1}^{k-1} s_i(u) + r_{k} \cdot \MMS{}_i(k) \le
        \sum_{u = 1}^{k} s_i(u),
    \end{align*}
    where the last inequality comes from the definition of the sequence $r_1,\dots,r_{k_j}$.
    However, this stands in a contradiction to the fact that all items from $\rhd_i(1),\dots,\rhd_i(k)$ are unavailable, hence allocated.

    \textit{Case (a.2)}.
    If $k^* < k$, then
    by \cref{claim:goods:efx:tracking}b,
    it must be that $s_i(k',A^{t-1}_i) = 0$, for every $k' > k^*$.
    Also, since some goods from $\rhd_i(k^*)$ are still available after iteration $t-1$, if it is agent $i$ that is chosen in line 3 of iteration $t$, it will receive a good from this class.
    Thus, also $s_i(k',A^{t}_i) = 0$, for every $k' > k^*$.

    Thus, we have,
    \begin{equation}
    \label{eq:goods:efx:mms:4}
    |A^t_i| = \sum_{u = 1}^{k^*} s_i(u, A^t_i) = \sum_{u = 1}^{k^*} \MMS{}_i(u),
    \end{equation}
    while for every $j \in N \setminus L \setminus \{i\}$,
    we have
    \begin{equation}
    \label{eq:goods:efx:mms:5}
    |A^t_j| \le |A^t_i| + 1 = 1 + \sum_{u = 1}^{k^*} \MMS{}_i(u).
    \end{equation}
    Adding equation \eqref{eq:goods:efx:mms:1} and inequalities \eqref{eq:goods:efx:mms:4} and \eqref{eq:goods:efx:mms:5} sidewise together, we get
    \begin{equation}
    \label{eq:goods:efx:mms:a:6}
    \sum_{j \in N} |A^t_j| \le
        \sum_{u = 1}^{k^*-1} s_i(u) + r_{k^*} \cdot \MMS{}_i(k^*) + (r_{k^*}-1).
    \end{equation}
    From the definition of the sequence $r_1,\dots,r_{k_j}$ we have that
    \begin{align*}
        \sum_{u=k^*}^k s_i(u) &=
        \sum_{u=k^*}^k ( r_u - r_{u+1} + r_u \cdot \MMS{}_i(u)) \\ &\ge
        r_{k^*} \cdot \MMS{}_i(k^*) + r_k \cdot \MMS{}_i(k) + \sum_{u=k^*}^k ( r_u - r_{u+1}) \\ &\ge
        r_{k^*} \cdot \MMS{}_i(k^*) + r_k \cdot \MMS{}_i(k) + r_{k^*} - r_k \\ &\ge
        r_{k^*} \cdot \MMS{}_i(k^*) + r_{k^*},
    \end{align*}
    where the last inequality comes from the fact that $s_i(k, A^t_i) < \MMS{}_i(k)$, thus $\MMS{}_i(k) \ge 1$.
    Combining this with inequality~\eqref{eq:goods:efx:mms:a:6}, yields
    \(
    \sum_{j \in N} |A^t_j| <
    \sum_{u=1}^k s_i(k),
    \)
    which contradicts the assumption that $\rhd_i(k')$ is unavailable, thus also allocated, for every $k' \in [k]$.

    
    \textit{Point (b).}
    Assume otherwise, i.e., there is $k \in \mathbb{N}$ such that $\rhd_i(1:k) \subseteq X^t$, agent $j \in U^t$ such that $|A^t_j| <|A^t_i|$,
    and it holds that
    $s_i(A^t_i) \le \MMS{}_i$ and
    $s_i(k',A^t_i) = 0$, for every $k' > k$.
    From (a) we know that this means that
    $s_i(1 \colon k,A^t_i) = \MMS{}_i(1 \colon k)$.

    Observe that since the point (b) was satisfied in iteration $t-1$, it must be that there is $k^* \in [k]$ such that $\rhd_i(k^*) \not \subseteq X^{t-1}$.
    Since, by Claim~\ref{claim:goods:efx:tracking}c, $s_i(A^{t-1}_i) \le s_i(A^t_i) \le \MMS{}_i$, this means that $i \in U^{t-1}$.
    Thus, from (c) for iteration $t-1$, we know that there is $L \subseteq N \setminus U^{t-1}$ such that $|L|=n-r_{k^*}$ and
    \begin{equation}
    \label{eq:goods:efx:mms:b:1}
    \sum_{h \in L} |A^t_h| =
        \sum_{u = 1}^{k^*-1} s_i(u) - r_{k^*} \cdot \sum_{u = 1}^{k^*-1} \MMS{}_i(u).
    \end{equation}
    
    Let us denote $s = |A^{t}_i|$.
    Then, $|A^{t}_j| < s$.
    Since $j \in U^{t}$,
    from \cref{claim:goods:efx:basic}b
    we know that $|A^{t}_j| = s-1$
    and no agent has more than $s$ goods.
    If $k^* < k$, then
    from \cref{claim:goods:efx:tracking}b
    it must be that $s_i(k',A^{t-1}_i) = 0$, for every $k' > k^*$.
    Hence, no matter the relation between $k^*$ and $k$, it holds that
    \[
        s = 
        \sum_{u=1}^k s_i(u, A^{t}_i) = 
        \sum_{u=1}^{k^*} s_i(u, A^{t}_i) =
        \sum_{u=1}^{k^*} \MMS{}_i(u).
    \]
    Therefore,
    \begin{equation}
    \label{eq:goods:efx:mms:b:2}
    |A^t_j| = s-1 = \sum_{u = 1}^{k^*} \MMS{}_i(u) -1,
    \end{equation}
    while for every $h \in N \setminus L \setminus \{j\}$,
    we have
    \begin{equation}
    \label{eq:goods:efx:mms:b:3}
    |A^t_h| \le s = \sum_{u = 1}^{k^*} \MMS{}_i(u).
    \end{equation}
    Adding inequalities \eqref{eq:goods:efx:mms:b:1}, \eqref{eq:goods:efx:mms:b:2}, and \eqref{eq:goods:efx:mms:b:3} sidewise together, we get
    \(
    \sum_{h \in N} |A^t_h| \le
        \sum_{u = 1}^{k^*-1} s_i(u) - r_{k^*} \cdot \sum_{u = 1}^{k^*-1} \MMS{}_i(u)
        + r_{k^*} \cdot \sum_{u = 1}^{k^*} \MMS{}_i(u) - 1,
    \)
    which gives
    \begin{align*}
    \sum_{h \in N} |A^t_h| &<
        \sum_{u = 1}^{k^*-1} s_i(u) + r_{k^*} \cdot \MMS{}_i(k^*) \\ &\le
        \sum_{u = 1}^{k^*} s_i(u).
    \end{align*}
    However, this stands in a contradiction to the fact that all items from $\rhd_i(1),\dots,\rhd_i(k^*)$ are unavailable, hence allocated.


    \textit{Point (c)}.
    Take iteration $t\in [m]$, agent $i \in U^t$, and class $k \in [k_i]$ such that
    $\rhd_i(1 \! : \! k) \subseteq X^t$ and
    $s_i(1 \! : \! k, A^t_i) = \MMS{}_i(1 \! :\! k)$.
    First, let us consider the case in which 
    the first $k$ classes of $i$
    are unavailable also in iteration $t-1$, i.e.,
    $\rhd_i(1 \colon k) \subseteq X^{t-1}$.
    Observe that $U$ does not increase (\cref{claim:goods:efx:basic}a) and
    the agents in $U$ do not receive more goods if $U$ is not empty
    (and since $i \in U^t$, it is not empty).
    Thus, since (c) holds for iteration $t-1$,
    we get that it holds also for $t$.
    Therefore, for the remainder of the proof,
    we will assume that there is $k^* \in [k]$
    such that $\rhd_i(k^*) \not \subseteq X^{t-1}$.

    From Claim~\ref{claim:goods:efx:basic}a and the assumption that $i \in U^t$ we get that $i \in U^{t-1}$.
    Thus, from (c) for iteration $t-1$ we know that $|U^{t-1}| \le r_{k^*}$ and there is $L \subseteq N \setminus U^{t-1}$ such that $|L| = n - r_{k^*}$ and
    \begin{equation}
    \label{eq:goods:efx:mms:c:1}
    \sum_{j \in L} |A^t_j| =
        \sum_{u = 1}^{k^*-1} s_i(u) - r_{k^*} \cdot \sum_{u = 1}^{k^*-1} \MMS{}_i(u).
    \end{equation}

    The fact that $\rhd_i(1 \colon k) \subseteq X^t$, implies that $\sum_{j \in N} |A^t_j \cap \rhd_i(1 \colon k)| = \sum_{u = 1}^k s_i(u)$.
    Combining this with equation~\eqref{eq:goods:efx:mms:c:1}
    implies that among agents in $N \setminus L$
    there are
    $\sum_{u=k^*}^k s_i(u) + r_{k^*} \cdot \sum_{u = 1}^{k^*-1} \MMS{}_i(u)$ goods from $\rhd_i(1 \colon k)$.
    From the definition of the sequence $r_1, \dots, r_{k_i}$, we have that
    \begin{align*}
        \sum_{u=k^*}^k s_i(u) &=
        \sum_{u=k^*}^k ( r_u - r_{u+1} + r_u \cdot \MMS{}_i(u)) \\ &=
        \sum_{u=k^*}^k r_{u} \cdot \MMS{}_i(u) + r_{k^*} - r_{k+1}.
    \end{align*}
    Thus,
    \[
        \sum_{j \in N \setminus L} |A^t_j \cap \rhd_i(1 \colon k)| = 
        r_{k^*} - r_{k+1} + r_{k^*} \cdot \sum_{u = 1}^{k^*} \MMS{}_i(u) +
        \sum_{u = k^* + 1}^{k} r_u \cdot \MMS{}_i(u).
    \]
    Since $\rhd_i(k^*) \not \in X^{t-1}$,
    then by Claim~\ref{claim:goods:efx:tracking}b,
    we know that $s_i(u,A^t_i) =0$ for every $u > k^*$.
    Hence, by the assumption that $s_i(1 \! : \! k, A^t_i) = MMS{}_i(1 \! : k)$,
    we get that $\MMS{}_i(u) = 0$ for every $u \in \{k^*, k^*+1, \dots, k\}$.
    Therefore,
    \begin{equation}
    \label{eq:goods:efx:mms:c:2}
        \sum_{j \in N \setminus L} |A^t_j \cap \rhd_i(1 \colon k)| = 
        r_{k^*} - r_{k+1} + r_{k^*} \cdot \sum_{u = 1}^{k^*} \MMS{}_i(u).
    \end{equation}

    Now, let us denote $s = |A^{t}_i| = \sum_{u=1}^{k_i} s_i(u, A^t_i) = \sum_{u = 1}^{k^*} \MMS{}_i(u)$.
    By \cref{claim:goods:efx:basic}b,
    no agent has more than $s+1 = 1 + \sum_{u = 1}^{k^*} \MMS{}_i(u)$ goods.
    By equation~\eqref{eq:goods:efx:mms:c:2},
    the fact that $|N \setminus L| = r_{k^*}$,
    and a pigeonhole principle kind of argument,
    this means that there are at least $r_{k^*} - r_{k+1}$
    agents with $s + 1$ goods from $\rhd_i(1 \colon k)$.

    Let $j$ be an arbitrary agent with $s+1$ goods from $\rhd_i(1 \colon k)$
    at the end of iteration $t$.
    We will show that necessarily $j \not \in U^t$.
    Observe that all goods of agent $j$ belong to $\rhd_i(1 \colon k)$,
    hence they are unavailable.
    In particular, if $l$ is the last indifference class such that $s_j(l,A^t_j) > 0$,
    then by Claim~\ref{claim:goods:efx:tracking}d we know that $\rhd_j(l) \subseteq X^t$.
    Also, from Claim~\ref{claim:goods:efx:tracking}b we know that
    this is the case for every indifference class of $j$ before $l$.
    Thus, since $i$ has less goods than $j$,
    from (b) we have that $s_j(A^t_j) >_{lex} \MMS{}_j$.
    Hence, if fairness criterion $F=\MMS{}$,
    then $j \not \in U^t$.
    
    Now, let us show that the same holds
    if fairness criterion $F=\EFX{}+\MMS{}$.
    We can show that $s_j(A^t_j) >_{lex} \MMS{}_j$ in the same way,
    hence let us focus on showing that there is no path from $j$ to $i$ in graph $G_{envy}(A^t,U^{t-1})$.
    Assume otherwise, i.e., there exists a path $(i_0,i_1,\dots,i_r)$ in $G_{envy}(A^t,U^{t-1})$ such that $j=i_0$ and $i=i_r$.
    Similarly to the proof of Claim~\ref{claim:goods:efx:equal}c,
    we will inductively show that this would imply
    that every agent in the path has $s+1$ unavailable goods,
    which would contradict the fact that $|A^t_i| = s$.
    The basis of induction we have already established, i.e.,
    agent $j$ has $s+1$ goods and all of them are unavailable.
    For the inductive step, assume that for some $u \in [r-1]$
    agent $i_u$ has $s+1$ unavailable goods.
    Since no agent has more than $s+1$,
    this means that all $i_u$'s goods are unavailable.
    Hence, if agent $i_u$ does not actually envy $i_{u+1}$,
    but only potentially,
    then $i_{u+1}$ has $s+1$ unavailable goods from Claim~\ref{claim:goods:efx:unavailable}.
    Thus, assume that $i_u$ actually envy agent $i_{u+1}$.
    Observe that $i_{u+1}$ has to be in $U^t$ if there is a potential envy path
    from $i_{u+1}$ to $i_r$ in $G_{envy}(A^t,U^{t-1})$.
    Moreover, since $i_u$ envies $i_{u+1}$, from Claim~\ref{claim:goods:efx:basic}d
    we know that $i_{u+1}$ is before $i_u$ in the ordering $\sigma$.
    Hence, by Claim~\ref{claim:goods:efx:basic}b, agent $i_{u+1}$ also has $s+1$ goods.
    But then, from Claim~\ref{claim:goods:efx:equal}c we get that all of the goods of $i_{u+1}$ are unavailable.
    Thus, by induction, we get that $i$ has $s+1$ unavailable goods, which is a contradiction.
    Hence, also in this case, $j \not \in U^t$.
    Let us denote by $L'$ the union of $L$ and the set of such agents in set $N \setminus L$.
    Observe that $|L'| = (n - r_{k^*}) + (r_{k^*} - r_{k+1}) = n - r_{k+1}$.
    Moreover, from the definition of the sequence $r_1,\dots,r_{k_i}$ and the fact that $\MMS{}_i(u)=0$ for every $u \in \{k^*+1, k^*+2,\dots, k\}$ we get that
    \begin{align*}
        \sum_{j \in L' \setminus L} |A^t_j| &= (r_{k^*} - r_{k+1}) \cdot (s+1)\\
        &= (r_{k^*} - r_{k+1}) \cdot (1+\sum_{u=1}^{k} \MMS{}_i(u))\\
        &= (r_{k^*} - r_{k+1})(1 + \MMS{}_i(k^*) + \sum_{u=1}^{k^*-1} \MMS{}_i(u))\\
        &= \sum_{u=k^*}^k s_i(u) - r_{k+1} \cdot \MMS{}_i(k^*) + (r_{k^*} - r_{k+1}) \cdot \sum_{u=1}^{k^*-1} \MMS{}_i(u).
    \end{align*}
    Hence, from equation~\eqref{eq:goods:efx:mms:c:1}, we get
    \begin{align*}
        \sum_{j \in L'} |A^t_j| &= \sum_{j \in L' \setminus L} |A^t_j| + \sum_{j \in L} |A^t_j| \\
        &= \sum_{u = 1}^{k} s_i(u) - r_{k+1} \cdot \sum_{u = 1}^{k^*} \MMS{}_i(u) \\
        &= \sum_{u = 1}^{k} s_i(u) - r_{k+1} \cdot \sum_{u = 1}^{k} \MMS{}_i(u).
    \end{align*}
    Therefore, $L'$ is a desired set of agents.
\end{proof}


\textbf{\underline{Part 3: Fairness guarantees.}}

In this part of the proof,
we will show how \cref{claim:goods:efx:tracking,claim:goods:efx:basic,claim:goods:efx:unavailable,claim:goods:efx:nonempty,claim:goods:efx:equal,claim:goods:efx:mms} imply the fairness guarantees in the statements (A)--(D).
Let us consider each of the statements separately.


\textbf{\underline{(A): EF1.}}
Let us show that if there is no criteria, i.e., $F = null$,
then  the output allocation is \EF{1}. 
To this end, observe that at the end of every iteration, 
$U$ contains all agents, i.e., for every $t \in [m], U^t = N$.

For iteration $t$ and two agents $i,j$ let $e^t_{i,j}$ be a vector such that
\[
    e^t_{i,j}(k) =
    \begin{cases}
        1, & \mbox{if } k \mbox{ is minimal s. t. } s_j(k,A^t_i)>0,\\
        0, & \mbox{otherwise.}
    \end{cases}
\]
Observe that an allocation $A^t$ is \EF{1}, if and only if, for every agents $i,j \in N$ it holds that $s_j(A^t_j) \ge_{lex} s_j(A^t_i) - e^t_{i,j}$.
Let us show that indeed, for every $t \in [m]$, allocation $A^t$ is \EF{1}.
For a contradiction, let us assume that there is $t \in [m]$ such that
$s_j(A^t_i) - e^t_{i,j} >_{lex} s_j(A^t_j)$
for some $i,j \in N$ and let us take the first such $t$.

Let $t' < t$ be the last iteration in which agent $j$ has less goods than in iteration $t$, i.e., $|A^{t'}_j| < |A^{t}_j|$.
This means that in iteration $t'+1$ agent $j$ is chosen in line 3 of the algorithm and receives a good, $g$, in line 5.
Let $k \in [k_j]$ be such that $g \in \rhd_j(k)$ and let $e_k$ be a vector with $1$ on $k$-th position and 0 otherwise.
Then, by \cref{claim:goods:efx:tracking}c,
\begin{equation}
    \label{eq:goods:ef1:1}
    s_j(A^{t}_j) = s_j(A^{t'}_j) + e_k.
\end{equation}

Now, let us compare vectors
$\mathbf{x} = (x(1),\dots,x(k_j)) = s_j(A^{t'}_i) + e_k - e^{t'}_{i,j}$ and $\mathbf{y} = (y(1),\dots,y(k_j))= s_j(A^{t}_i) - e^t_{i,j}$.
Observe that:
\begin{itemize}
    \item[(a)] $|A^t_j| \ge \sum_{u=1}^{k_j} y(u)$ by Claim~\ref{claim:goods:efx:basic}b;
    \item[(b)] $x(u) \ge y(u)$ for every $u < k$, since,
    by \cref{claim:goods:efx:tracking}b,
    all goods from the first $k-1$ classes of $j$ are unavailable $\rhd_j(1:k-1) \subseteq X^t$, which means that between iteration $t'$ and $t$, agent $i$ can only gain new goods from the $k$-th and further indifference classes of $j$;
    \item[(c)] $s_j(u, A^t_j) = 0$ for every $u \in [k]$, by Claim~\ref{claim:goods:efx:tracking}b and the fact that $g \in \rhd_j(k)$ is available in iteration $t-1$; and finally
    \item[(d)] $s_j(A^t_j) \ge_{lex} \mathbf{x}$.
    This holds since from minimality of $t$, we have
    $s_j(A^{t'}_j) \ge s_j(A^{t'}_i) - e^{t'}_{i,j}$.
    Combining this with equation~\eqref{eq:goods:ef1:1}, yields
    $s_j(A^{t}_j) \ge_{lex} s_j(A^{t'}_i) + e_k - e^{t'}_{i,j} = \mathbf{x}$.
\end{itemize}

Finally, let us show that $s_j(A^t_j) \ge_{lex} \mathbf{y}$.
If there is $u \in [k-1]$ such that either
$x(u) > y(u)$ or
$s_j(u,A^t_j) > x(u)$,
then by (b) and (d) we get that $s_j(A^t_j) >_{lex} \mathbf{y}$.
Hence, assume otherwise, i.e., $s_j(u, A^t_j) = x(u) = y(u)$ for every $u \in [k-1]$.
Then, from (c) and (a) we get that
\[
 s_j(k, A^t_j) =
 |A^t_j| - \sum_{u=1}^{k-1} s_j(u, A^t_j) \ge
 \sum_{u=1}^{k_j} y(u) - \sum_{u=1}^{k-1} y(u) \ge
 y(k),
\]
where the equality is only possible if $y(u)=0$ for every $u > k$.
Thus, $s_j(A^t_j) \ge_{lex} \mathbf{y}$,
which contradicts the assumption that $i$ and $j$ violate \EF{1} at the end of iteration $t$.


\textbf{\underline{(B): \EFX{}.}}
Let us show that if criteria $F=EFX$,
then the output allocation satisfies \EFX{}.

Assume by contradiction
that for some iteration $t \in \{0,\dots,m\}$
allocation $A^t$ is not \EFX{}
and let us take the first such $t$.
Since for $t = 0$,
empty allocation $A^0$ is \EFX{},
we know that $t > 0$.
This means that there exist
agents $i,j \in N$
and good $g \in A^t_i$
such that $j$ prefers $A^t_i \setminus \{g\}$
over its bundle.
From \cref{claim:goods:efx:basic}d we know that
$i$ is before $j$ in ordering $\sigma$.

Let $t' < t$ be the last iteration
in which agent $i$ had smaller number of goods, i.e.,
$|A^{t'}_i| < |A^{t}_i|$.
This means that in iteration $t' +1$ agent $i$
was chosen in line 4 of the algorithm.
Thus, by Claim~\ref{claim:goods:efx:nonempty}, $i \in U^{t'}$ and
$|A^{t'}_i| = |A^{t'+1}_i| - 1 = |A^{t}_i| - 1$.
Observe that $j \in U^t$ as otherwise
there would be some iteration $t' < t$
in which $j$ did not potentially envy $i$,
so, by \cref{claim:goods:efx:tracking}c,
$j$ could not envy $i$ in iteration $t$.
Hence, since $i$ is before $j$ in ordering $\sigma$,
from \cref{claim:goods:efx:basic}b
we get that in iteration $t'$
agents $i$ and $j$ had equal number of goods, i.e.,
$|A^{t'}_i| = |A^{t'}_j|$.
Let us show that necessarily $t' + 1 \neq t$.
Assume otherwise.
By \cref{claim:goods:efx:equal} and the fact that $i \in U^{t'}$
this means that $j$ does not envy $i$ in iteration $t'$.
Let $k^* \in \mathbb{N}$ be the last position
for which there is a positive number in $j$'s bundle score, i.e.,
$s_j(k^*, A^{t'}_j) > 0$.
By Claim~\ref{claim:goods:efx:tracking}c, $s_j(k^*, A^{t'+1}_j) > 0$ as well.
Then, a good, $g'$, that agent $i$ receives in iteration $t' +1$
has to be in a $k$-th indifference class of $j$ for some $k < k^*$
(otherwise removing any good from
$A^{t'+1}_i$ would eliminate the envy).
However, since $g'$ was available in iteration $t'$,
this contradicts \cref{claim:goods:efx:tracking}b.

Therefore, $t' + 1 < t$ and
$j$ and $i$ do not violate \EFX{} in iteration $t' + 1$.
Since $|A^t_i| = |A^{t'+1}_i|$
and, by \cref{claim:goods:efx:tracking}c, $A^t_j \succeq_j A^{t'+1}_j$,
this means that between iteration $t$ and $t'+1$
agent $i$ exchanged one of its goods for a good $g$
that is preferred by $j$ over one of goods from $A^{t'+1}_j$.
However, this means that $g$ was available at iteration $t'+1$,
which contradicts \cref{claim:goods:efx:tracking}b.
Thus, $A^t$ is \EFX{} for every $t \in [m]$.


\textbf{\underline{(C): \MMS{}.}}
Observe that if criteria $F=MMS$,
then the fact that the output allocation satisfies \MMS{}
follows directly from \cref{claim:goods:efx:mms}a:
at the end of the algorithm all goods are unavailable, thus $s_i(A_i) \ge_{lex} \MMS{}_i$, for every agent $i \in N$.


\textbf{\underline{(D): \EFX{}+\MMS{}.}}
Let us show that if criteria $F=\EFX{}+\MMS{}$,
then the output allocation satisfies \EFX{} and \MMS{}.
Observe that \MMS{} follows directly from \cref{claim:goods:efx:mms}a:
at the end of the algorithm all goods are unavailable, thus $s_i(A_i) \ge_{lex} \MMS{}_i$, for every agent $i \in N$.

Therefore, let us focus on proving \EFX{}.
Assume by contradiction
that for some iteration $t \in \{0,\dots,m\}$
allocation $A^t$ is not \EFX{}
and let us take the first such $t$.
Since for $t = 0$,
empty allocation $A^0$ is \EFX{},
we know that $t > 0$.
This means that there exist
agents $i,j \in N$
and good $g \in A^t_i$
such that $j$ prefers $A^t_i \setminus \{g\}$
over its bundle.
From \cref{claim:goods:efx:basic}d we know that
$i$ is before $j$ in ordering $\sigma$.

Let $t' < t$ be the last iteration
in which agent $i$ had smaller number of goods, i.e.,
$|A^{t'}_i| < |A^{t}_i|$.
This means that in iteration $t' +1$ agent $i$
was chosen in line 3 of the algorithm.
Thus, by \cref{claim:goods:efx:nonempty},
$i \in U^{t'}$ and
$|A^{t'}_i| = |A^{t'+1}_i| - 1 = |A^{t}_i| - 1$.
Observe that $j \in U^t$ as otherwise
there would be some iteration $t'' < t$
in which $j$ did not potentially envy $i$,
so, by \cref{claim:goods:efx:tracking}c,
$j$ could not envy $i$ in iteration $t$.
Hence, since $i$ is before $j$ in ordering $\sigma$,
from \cref{claim:goods:efx:basic}b
we get that in iteration $t'$
agents $i$ and $j$ had equal number of goods, i.e.,
$|A^{t'}_i| = |A^{t'}_j|$.

Let us prove that $j$ does not envy $i$ in iteration $t'$.
Assume otherwise and observe that, by Claim~\ref{claim:goods:efx:equal}c,
all goods of $j$ are unavailable in iteration $t'-1$.
Hence, by Claim~\ref{claim:goods:efx:tracking}d,
also all indifference classes of $j$ containing these goods are unavailable.
Therefore,
from Claim~\ref{claim:goods:efx:mms}b,
$A^{t'-1}_i >_{lex} \MMS{}_i$.
However, combining this with Claim~\ref{claim:goods:efx:equal}d yields $j \not \in U^{t'-1}$,
which is a contradiction.

Now, we will show that necessarily $t'+1 \neq t$.
Let us denote the good that agent $i$ receives in line 6 of iteration $t'+1$ as $g'$. 
Also, let $k^* \in \mathbb{N}$ be the last position
for which there is a positive number in $j$'s bundle score, i.e.,
$s_j(k^*, A^{t'}_j) > 0$.
By \cref{claim:goods:efx:tracking}b,
$g' \in \rhd_j(k^*:k_j)$.
Since by \cref{claim:goods:efx:tracking}c
$s_j(A^{t'}_j)=s_j(A^{t'+1}_j)$,
this means that removing any good from $A^{t'+1}_i$
would eliminate the envy.
Thus, $i$ and $j$ do not violate \EFX{} in iteration $t'+1$, hence $t'+1 \neq t$.

Therefore, $t' + 1 < t$ and
$j$ and $i$ do not violate \EFX{} in iteration $t' + 1$.
Since $|A^t_i| = |A^{t'+1}_i|$
and, by \cref{claim:goods:efx:tracking}c, $A^t_j \succeq_j A^{t'+1}_j$,
this means that between iteration $t$ and $t'+1$
agent $i$ exchanged one of its goods for a good $g$
that is preferred by $j$ over one of goods from $A^{t'+1}_j$.
However, this means that $g$ was available at iteration $t'+1$,
which contradicts \cref{claim:goods:efx:tracking}b.
Thus, $A^t$ is \EFX{} for every $t \in [m]$.


\textbf{\underline{Part 4: Pareto optimality.}}

\noindent
Finally, let us prove that the output of our algorithm is always a \PO{} allocation.

To this end, we will show that for each iteration $t \in \{0,\dots,m\}$,
partial allocation $A^t$ is \PO{}.
Assume otherwise.
Let us take the first iteration $t$
such that $A^t$ is not \PO{}.
From \cref{thm:po:aziz}
we know that
there is an alternating path,
$p = (g_0, i_1, g_1, \dots, i_s, g_s)$ and agent $i_0$ such that
$(g_s, i_0) \in A^t$ and
$\psi(i_0, g_0) < \psi(g_s,i_0)$.
Since allocation $A^{t-1}$ is \PO{},
either $(g_s, i_0) \not \in A^{t-1}$ or
path $p$ is not alternating for $A^{t-1}$.
Either way, this means that
at least one of the goods on path $p$
changed the owner in iteration $t$.
Let us take the one such good $g_u$
with the smallest $u \in [s]$.
As $g_u$ changed the owner,
it has to be available in $t-1$,
i.e., $g_u \not \in X^{t-1}$.
Since $u$ is the smallest,
the part of the original path,
$(g_0, i_1, g_1, \dots, i_u, g_u)$,
is still an alternating path in $A^{t-1}$.
Thus, $g_0$ is available in $t-1$ as well.
By \cref{claim:goods:efx:tracking}b,
this means that agent $i_0$ weakly prefers
each of its goods in $A^t$ over $g_0$.
However, this stands in contradiction to the fact that
$(g_s, i_0) \in A^t$ and
$\psi(i_0, g_0) < \psi(g_s,i_0)$.
\end{proof}

\section[Comparison with Aziz et al. 2023]{Comparison with \citetapp{aziz2023possible}}
\label{app:comparison}

In this appendix,
we compare \cref{alg:goods:efx} and the algorithm developed by \citetapp{aziz2023possible}
for the purpose of studying possible fairness
that also achieves \EFX{}, \MMS{} and \PO{} for every weakly lexicographic goods-only instance.
We note that both algorithms were developed independently and there are important differences between them both in terms of the techniques they utilize and the final outcomes they produce (see \cref{ex:goods:difference}).

Our algorithm (in contrast to \citet{aziz2023possible}'s) directly checks a predefined fairness requirement
for the bundle of each agent in each round of the runtime.
Then, the agents with an unfair advantage stop 
receiving more goods
and the priority is given to the worse-off agents.
Such construction has several benefits.

Firstly, the working mechanism behind our 
algorithm is well-structured and easy to understand, 
which is in a sharp contrast with the
more black-box approach of \citet{aziz2023possible}.
Thus, we believe that our formulation of the algorithm 
holds an educational value as it allows for explanation of where the fairness comes from.
    
Secondly, our algorithm is customizable in the sense that
a set of imposed fairness requirements can be modified, which will directly affect 
the outcome of the algorithm.
This in turn 
allows for easy modification of the algorithm to different notions of fairness
(including and beyond the notions studied in this paper).

Finally, the construction of our algorithm introduces the techniques of potential envy and preference graph, which can be found useful on their own.

In contrast, the algorithm of \citetapp{aziz2023possible} uses an exchange graph, which is a directed graph whose vertices are goods and there is an edge $(g, g')$ for $g,g' \in M$ if there is an agent $i \in N$ such that $g \in A_i$ and $g' \succeq_i g$.
The algorithm starts by initializing an empty bundle for each agent.
The agents that can still receive goods are kept in set $N$, which is updated over the run of the algorithm.
In each round, each agent in $N$ (in an arbitrary order) receives its most preferred good out of the unallocated ones or the ones that can be freed if other agents exchange their goods in along the edges of the exchange graph.
After the goods are allocated,
the set $N$ is updated.
To this end, the algorithm looks at the set of goods distributed in this round.
If some agent, $i$, strictly prefers the good received by another agent, $j$, over its own,
then agent $j$ is removed.
Moreover, whenever an agent is removed,
all agents with whom it can exchange goods received in this round are also removed.
The algorithm proceeds this way until all of the goods are allocated.

In addition to the differences in the construction,
the outcomes of the two algorithms can also be significantly different.
Let us consider the following example to illustrate this.

\begin{example} 
    \label{ex:goods:difference}
    Consider the instance given in \cref{ex:goods:efx}.

    \begin{align*}
        1 &: \quad \{g_1, {\tikz[baseline=(char.base)]{\node[shape=circle,draw,inner sep=1pt,line width=0.5pt,color=red] (char) {$\textcolor{black}{\underline{g_2}}$};}}\} \rhd \{g_3, {\tikz[baseline=(char.base)]{\node[shape=circle,draw,inner sep=1pt,line width=0.5pt,color=red] (char) {$\textcolor{black}{g_4}$};}} \} \\
        2 &: \quad {\tikz[baseline=(char.base)]{\node[shape=circle,draw,inner sep=1pt,line width=0.5pt,color=red] (char) {$\textcolor{black}{\underline{g_1}}$};}} 
        \rhd \{g_2, g_3, g_4\} \\
        3 &: \quad g_1 \rhd \{g_2, {\tikz[baseline=(char.base)]{\node[shape=circle,draw,inner sep=1pt,line width=0.5pt,color=red] (char) {$\textcolor{black}{\underline{g_3}}$};}},
        \underline{g_4}\} 
    \end{align*} 
    
    \underline{The outcome of 
    \cref{alg:goods:efx} with $F = \EFX{} + \MMS{}$:} 

    In iteration $1$, agent $1$ 
    receives $g_1$ (see the red highlight in \cref{fig:goods:ex:g_pref}). 
    Observe that good $g_1$ is still available since 
    there exists an alternating path $(g_1, 1, g_2)$ 
    ending in an unallocated good.
    Thus, in iteration $2$, agent $1$ exchanges $g_1$ 
    for $g_2$ and agent $2$ receives $g_1$.
    Then, in the potential envy graph 
    at the end of iteration $2$, i.e., $G_{envy}(A^2, U^1)$,
    agent $3$ is the only agent that forms
    a strongly connected component 
    with no incoming edge, $U_e^2$. 
    Moreover, since for each of agents $1$ and $2$,
    its bundle is above the respective 
    \MMS{} threshold, we remove them from $U$, 
    thus $U^2 = \{3\}$. 
    Therefore, agent $3$ receives goods $g_3$ and $g_4$.  
    The final allocation is underlined.

    \underline{The outcome of the algorithm of
    \citetapp{aziz2023possible}:}
    
    In the first round,
    the algorithm begins by giving
    agent $1$ good $g_1$.
    Next, in the same round,
    since agent $1$ can exchange good $g_1$ for $g_2$, agent $2$ gets good $g_1$ and agent $1$, $g_2$.
    The first round ends,
    with agent $3$ receiving good $g_3$.
    Observe that agent $3$ strictly prefers good $g_1$ over $g_3$, thus agent $2$
    is removed from set $N$, 
    i.e., $N = \{1,3\}$. 
    In the second round, agent $1$ receives good $g_4$.
    The final allocation
    is circled with red.
    
    Observe that although the allocation
    produced by the algorithm of \citetapp{aziz2023possible} is \EFX{}, 
    agent $3$ envies agent $1$.
    In our \cref{alg:goods:efx},
    since agent $1$ received a good from its first indifference class, we know that it will never envy agent $3$, thus we do not give any more goods to agent $1$.
    As a result,
    agent $3$ does not envy agent $1$ and the total number of envy relations between the agents is smaller in the output of our algorithm.
\end{example}

\section{Examples of Running Algorithm 1}
\label{app:goods:examples}

In this section, we analyze several examples 
for the goods-only setting. 

In
\cref{ex:goods:ef1,ex:goods:efxnotmms,ex:goods:mms}, 
we illustrate how  
\cref{alg:goods:efx} runs
with different fairness criteria $F$ in its input
on the same instance.

\begin{example} 
    \label{ex:goods:ef1}
    Consider a goods-only instance with three agents, 
    seven goods and preferences as follows. 
    Each agent partitions the goods into 
    two indifference classes. 
    \begin{align*}
        1 &: \quad \{\underline{g_1}, g_2, \underline{g_3}\} \rhd \{g_4, g_5, g_6, \underline{g_7}\} \\
        2 &: \quad \{g_1, \underline{g_2}, g_3\} \rhd \{g_4, g_5, \underline{g_6}, g_7\} \\
        3 &: \quad \{g_1, g_2, g_3, \underline{g_4}, \underline{g_5}\} \rhd \{g_6, g_7\}
    \end{align*}
    
    \underline{\cref{alg:goods:efx} with $F = null$ 
    (\EF{1} and \PO{}):}

    First, agents $1$, $2$ and $3$ receive 
    goods $g_1$, $g_2$ and $g_3$, respectively. 
    Then, agent $1$ points to $g_3$, and we update 
    the allocation along the alternating path 
    $p = (g_3, 3, g_4)$ and assign $g_3$ to agent $1$.
    In the following iterations, 
    agents $2$ and $3$ receive 
    goods $g_6$ and $g_5$ in order. 
    Finally, agent $1$ receives good $g_7$.
    The final allocation is underlined.

\end{example}
\begin{example} 
    \label{ex:goods:efxnotmms}
    Consider the instance given 
    in \cref{ex:goods:ef1}.
    \begin{align*}
        1 &: \quad \{\underline{g_1}, g_2, \underline{g_3}\} \rhd \{g_4, g_5, g_6, g_7\} \\
        2 &: \quad \{g_1, \underline{g_2}, g_3\} \rhd \{g_4, g_5, g_6, g_7\} \\
        3 &: \quad \{g_1, g_2, g_3, \underline{g_4}, \underline{g_5}\} \rhd \{\underline{g_6}, \underline{g_7}\}
    \end{align*}
    
    \underline{\cref{alg:goods:efx} with $F = EFX$ 
    (\EFX, \PO{}, but not \MMS{}):}

    First, agents $1$, $2$ and $3$ receive 
    goods $g_1$, $g_2$ and $g_3$, respectively. 
    Observe that every agent potentially 
    envies every other agent 
    at the current allocation. 
    Thus, agent $1$ points to $g_3$, and we update 
    the allocation along the alternating path 
    $p = (g_3, 3, g_4)$ and assign $g_3$ to agent $1$.
    Then, agent $3$ is the only agent 
    that is not potentially envied. 
    Therefore, $3$ receives the rest of the goods. 
    The final allocation is underlined. 
    Observe that this allocation is \EFX{}. 
    However, agent $2$'s \MMS{} threshold is not met, 
    as it receives no good from 
    its second indifference class. 

\end{example}

\begin{example} 
    \label{ex:goods:mms}
    Consider the instance given 
    in \cref{ex:goods:ef1}.
    \begin{align*}
        1 &: \quad \{\underline{g_1}, g_2, \underline{g_3}\} \rhd \{g_4, g_5, g_6, g_7\} \\
        2 &: \quad \{g_1, \underline{g_2}, g_3\} \rhd \{g_4, g_5, \underline{g_6}, \underline{g_7}\} \\
        3 &: \quad \{g_1, g_2, g_3, \underline{g_4}, \underline{g_5}\} \rhd \{g_6, g_7\}
    \end{align*}
    
    \underline{\cref{alg:goods:efx} with $F = MMS$ 
    (\MMS{} and \PO{}):}

    First, agents $1$, $2$ and $3$ receive 
    goods $g_1$, $g_2$ and $g_3$, respectively.  
    Then, agent $1$ points to $g_3$, and we update 
    the allocation along the alternating path 
    $p = (g_3, 3, g_4)$ and assign $g_3$ to agent $1$.
    Agent $1$ prefers its bundle to 
    its \MMS{} threshold. However, 
    agents $2$ and $3$ have not yet 
    met their \MMS{}.  
    Thus, agent $2$ receives goods 
    $g_6$ and $g_7$, 
    while agent $3$ gets $g_5$. 
    The final allocation is underlined.

\end{example}


\section{EFX and MMS Relation}
\label{app:efx:mms:relation}

In the current section, we present 
examples for the relation between
\EFX{} and \MMS{} under both 
goods-only and chores-only settings.

The fact that \EFX{} (even with \PO{})
does not necessarily imply 
\MMS{} in the goods-only setting is visible in \cref{ex:goods:efxnotmms}. 
In contrast, as shown in \cref{subsec:chores:efxtomms},
\EFX{} implies \MMS{} for chores-only instances. 
However, \cref{ex:chores:mmsponotefx}
illustrates that 
\MMS{} coupled with \PO{} does not imply \EFX{}
for this setting, even when 
agents have at most two indifference classes. 

\begin{example} 
    \label{ex:chores:mmsponotefx}
    Consider a chores-only instance 
    with three agents, 
    three chores and preferences as follows. 
    Agent 3 has all chores as one class while 
    the others partition the chores into two indifference classes. 
    \begin{align*}
        1 &: \quad c_1 \rhd \{\underline{c_2}, \underline{c_3}\} \\
        2 &: \quad c_1 \rhd \{c_2, c_3\} \\
        3 &: \quad \{\underline{c_1}, c_2, c_3\}
    \end{align*}
     
    Observe that the underlined allocation is 
    \MMS{} and \PO{}. 
    But agent $1$ envies agent $2$
    more than \EFX{}. 
\end{example}

\section[Correctness of Algorithm 3]{Correctness of Algorithm~\ref{alg:chores:ef1}}
\label{app:chores:ef1}
\setcounter{claim}{0}
In this appendix, we show that
our \cref{alg:chores:ef1} indeed always finds an \EF{1} and \PO{} allocation.

Helpful for that will be
the observation made by \cite[Lemma 4.14]{ebadian2022fairly}
that the PO characterization in \cref{thm:po:aziz}
can be also adapted to chores-only instances.
For completeness, we provide its proof in \cref{app:chores:po}.

\begin{restatable}{theorem}{thmchorespo}\label{thm:po:ebadian}
Given a weakly lexicographic
    chores-only instance $(N,M,\rhd)$,
    an (possibly partial) allocation $A$ is PO,
    if and only if,
    there is no alternating path,
    $p = (c_0, i_1, c_1, \dots, i_s, c_s)$ and agent $i_0$ such that
    $(c_s, i_0) \in A$ and
    $\psi(i_0, c_0) > \psi(c_s,i_0)$.
\end{restatable}

Let us now move to the proof of the correctness of \cref{alg:chores:ef1}.

\thmchoresefone*

\begin{proof}
For every $t \in [m]$,
by $A^t$, we will denote the allocation
we have at the end of the 
$t$-th iteration of
the main loop of the algorithm (lines 2--7).
By $A^0$, we denote the initial empty allocation.

The proof is split in two parts
in which we focus on showing that the output allocation
is \EF{1} (part 1) and \PO{} (part 2).

\textbf{\underline{Part 1: The output allocation is EF1.}}

Let $e_k$ denote a vector with all zeros
except for 1 in the $k$-th position.
Also, let $X^t$ be the set of all unavailable
chores at the end of iteration $t$,
i.e., allocated chores from which
there is no alternating path
ending in an unallocated chore.
Let us show two claims that correspond to
\cref{claim:goods:efx:tracking}a and
\cref{claim:goods:efx:tracking}b
from the proof of \cref{thm:goods:efx}.
First, let us show that unavailable chores
do not become available.
\begin{claim_app_EF1}
    \label{claim:chores:ef1:available}
    For every iterations $t,t' \in [m]$
    such that $t < t'$, it holds that
    $X^t \subseteq X^{t'}$.
\end{claim_app_EF1}
\begin{proof}
    The proof is almost identical to that of
    \cref{claim:goods:efx:tracking}a in the proof of \cref{thm:goods:efx}.
    It suffices to show that $X^{t-1} \subseteq X^t$,
    for every $t \in [m]$.
    Take an arbitrary chore $c \in X^{t-1}$.
    Observe that every chore $c'$ such that
    there is an alternating path starting in $c$ and ending in $c'$ has to be unavailable.
    Thus, all of these chores
    will not change the ownership in iteration $t$.
    Thus, they will remain unavailable,
    i.e., $c \in X^t$.
\end{proof}

The second claim ensures that agents
weakly prefer their chores over the available ones.

\begin{claim_app_EF1}
    \label{claim:chores:ef1:best}
    For every agent $i \in N$,
    iteration $t \in [m]$,
    and chore $c \in A^t_i$,
    it holds that
    $c \succeq_i c'$,
    for every $c' \in M \setminus X^{t-1}$.
\end{claim_app_EF1}
\begin{proof}
    The proof is similar to that of
    \cref{claim:goods:efx:tracking}b
    in the proof of \cref{thm:goods:efx}.
    Assume by contradiction
    that there exists 
    $c \in M \setminus X^{t-1}$
    such that $c \succ_i c'$
    for some $c' \in A_i^t$.
    Let $k$ be such that $c' \in \rhd_i(k)$
    and let $t'\in [m]$ be the first iteration
    in which $i$ received a chore, $c''$,
    from its $k$-th or earlier indifference class.
    Thus, $t' \le t$.
    Since all chores in $A_i^{t'-1}$
    are preferred over chore $c''$,
    agent $i$ could not obtain it in line 5
    of the algorithm,
    by the definition of the alternating path.
    Hence, $i$ received it in line 6.
    However, by \cref{claim:chores:ef1:available},
    chore $c$ was also available at the end of iteration $t'-1$.
    Therefore, since $c \succ_i c' \succeq_i c''$,
    agent $i$ should receive $c$
    instead of $c''$---a contradiction.
\end{proof}

Next,
let us show how the score vectors
of agents are changing as the algorithm progresses.

\begin{claim_app_EF1}
    \label{claim:chores:ef1:scores}
    For every agent $i \in N$
    and iteration $t \in [m]$
    it holds that
    \[
        s_i(A_i^{t}) =
        \begin{cases}
            s_i(A_i^{t-1}), \quad \mbox{if i is not chosen in line 3 of iteration t} \\
            s_i(A_i^{t-1}) - e_k, \quad \mbox{otherwise} \\
        \end{cases}
    \]
    where $k \in [k_i]$ is the number of
    the last indifference class of $i$
    with available chores, i.e.,
    $s_i(k, M \setminus X^t) > 0$.
\end{claim_app_EF1}
\begin{proof}
First, consider iteration $t \in [m]$
such that $i$ is not chosen in line 3.
Then, the bundle of $i$ can change
only due to an alternating path.
By definition of an alternating path
agent $i$ cannot get worse off, i.e.,
$s_i(A_i^{t}) \ge_{lex} s_i(A_i^{t-1})$.
Assume by contradiction that
$s_i(A_i^{t}) >_{lex} s_i(A_i^{t-1})$.
This means that there is a chore
$c \in A_i^{t-1} \setminus A_i^{t}$
that got exchanged for
$c' \in A_i^t \setminus A_i^{t-1}$
such that $c' \succ_i c$.
This means that $c'$ was available
at the end of iteration $t-1$,
so, by \cref{claim:chores:ef1:available} also $t-2$.
However, this contradicts \cref{claim:chores:ef1:best}.
Hence, $s_i(A_i^{t}) = s_i(A_i^{t-1})$.

Now, if in iteration $t \in [m]$
agent $i$ is chosen in line 3,
then in line 6 it receives 
an available chore from its
last possible indifference class.
The thesis follows.
\end{proof}

Now, let us move to the main observation of the proof
that is \cref{claim:chores:ef1:just_before}.
Intuitively, the following claim guarantees that an agent
that is just about to pick a chore, $i$,
does not envy any other agent.
However, it is even a bit stronger than that.
Denote by $k$ the first indifference class of agent $i$
from which $i$ has some chores, i.e., $s_i(k,A^t_i)>0$.
Then, the claim guarantees that agent $i$
would not envy other agents,
even if we move all allocated chores from 
the first $k$ indifference classes of $i$ to the $k$-th class of $i$.
For this purpose, some additional notation will be helpful.
For every agent $i \in N$,
indifference class $k \in [k_i]$,
and a subset of chores $X \subseteq M$,
let us denote the score vector of $i$
with the first $k$-th positions merged as
\(
    s^k_i(X) = (\sum_{u \in [k]} s_i(u),
        s_i(k+1), \dots,
        s_i(k_i - k + 1)).
\)

\begin{claim_app_EF1}
    \label{claim:chores:ef1:just_before}
    For every agent $i \in N$,
    iteration $t \in [m]$,
    and $k \in [k_i]$
    such that agent $i$ is chosen in line 3
    of the algorithm in iteration $t+1$
    and $k$ is the first indifference class of $i$
    such that $s_i(k, A^t_i) > 0$
    or $k=k_i$ if $s_i(A^t_i) = (0,\dots,0)$,
    it holds that
    \(
        s_i^k(A_i^t) \ge_{lex} s_i^k(A_j^t),
    \)
    for every agent $j \in N$.
\end{claim_app_EF1}
\begin{proof}
    Assume otherwise and take the first iteration $t$
    for which there exist agents $i, j \in N$
    such that $i$ is chosen in line 3 of the algorithm
    in iteration $t+ 1$,
    and 
    \(
        s_i^k(A_j^t) >_{lex} s_i^k(A_i^t).
    \)

    Observe that $t > n$
    (otherwise, agent $i$ has no chores
    at the end of iteration $t$).
    Let $t' = t - n$.
    Note that agent $i$ is chosen in line 3 of the algorithm in iteration $t'+1$ and,
    by \cref{claim:chores:ef1:scores}, receives a chore from its $k$-th indifference class in line 6.
    Thus, from the minimality of $t$,
    we get that
    \(
        s_i^{k'}(A_i^{t'}) \ge_{lex} s_i^{k'}(A_j^{t'}),
    \)
    where $k'$ is the first indifference class such that
    $s_i(k',A^{t'}_i) > 0$.
    Since $k' \ge k$, this implies that
    \begin{equation}
        \label{eq:claim:chores:ef1:just_before}
        s_i^{k}(A_i^{t'}) \ge_{lex} s_i^{k}(A_j^{t'}).
    \end{equation}

    Since in iteration $t'+1$ agent $i$ gets a chore from its $k$-th class,
    all of the chores from its later classes are no longer available.
    Thus, the chores from these classes cannot change owners between iteration $t'$ and $t$.
    In particular, agent $j$ has the same chores
    from these classes in $t'$ and $t$.
    Since in line 3 of some iteration
    between $t'+1$ and $t$
    agent $j$ receives an additional chore,
    which has to be
    in the first $k$ indifference classes of $i$,
    we get that
    \(
        s_i^{k}(A_j^{t}) = s_i^{k}(A_j^{t'}) - e_1.
    \)
    Therefore, by \cref{claim:chores:ef1:scores} and equation~\eqref{eq:claim:chores:ef1:just_before},
    \[
        s_i^k(A_i^t) =
        s_i^{k}(A_i^{t'}) - e_1 \ge_{lex}
        s_i^{k}(A_j^{t'}) - e_1 =
        s_i^{k}(A_j^{t}),
    \]
    which is a contradiction.
\end{proof}

Now, let us prove the main theorem.
Fix arbitrary $i,j \in N$.
We will show by induction
that for every iteration $t \in [m]$
it holds that
$s_i(A_i^t) + e_k \ge_{lex} s_i(A_j^t)$,
where $k$ is the first indifference class of $i$
such that $s_i(k, A^t_i) > 0$,
or $k=k_i$, if $s_i(A^t_i)=(0,\dots,0)$.
Since there exists a chore $c \in A_i^t$ such that $s_i(A_i^t \setminus \{c\}) = s_i(A^t_i) + e_k$,
this will imply that $A^t$ is \EF{1}.
Clearly, for $t=0$ the thesis holds.
Assume that the thesis holds for some $t \in [m]$
and consider iteration $t + 1$.
Let $k$ be the first indifference class such that $s_i(k, A^{t+1}_i) >0$ or let $k=k_i$ if $i$ does not have chores.
By \cref{claim:chores:ef1:best},
the chores from 
the later than $k$-th indifference classes
of agent $i$ are not available.
Hence,
agent $j$ has the same chores
from these classes in $t$ and $t+1$.
Thus,
\begin{equation}
    \label{eq:chores:ef1}
    s^k_i(A^{t+1}_i) + e_k = s_i^k(A_j^t) \ge_{lex} s_i^k(A_j^{t+1})
\end{equation}
(the inequality is only possible if $j$
received an additional chore in line 6
of the iteration $t+1$).

If $i$ was chosen in line 3 of iteration $t+1$,
then, by \cref{claim:chores:ef1:just_before},
we have
\(
    s_i^k(A_i^t) \ge_{lex} s_i^k(A_j^t).
\)
Let $c$ be a chore
that $i$ received in line 6 of iteration $t$.
Then, by \cref{claim:chores:ef1:scores}
and \cref{eq:chores:ef1},
\[
    s_i^k(A_i^{t+1} \setminus \{c\}) =
    s_i^k(A_i^t) \ge_{lex}
    s_i^k(A_j^t) \ge_{lex}
    s_i^k(A_j^{t+1}).
\]
Since $s_i(k',A_i^{t+1} \setminus \{c\})=0$
for every $k' < k$,
this implies the inductive thesis.

Thus, let us assume that $i$ was
not chosen in line 3 in iteration $t+1$.
Then, by \cref{claim:chores:ef1:scores},
\cref{eq:chores:ef1},
and the inductive assumption,
\[
    s_i^k(A_i^{t+1}) + e_k =
    s_i^k(A_i^t) + e_k \ge_{lex}
    s_i^k(A_j^t) \ge_{lex}
    s_i^k(A_j^{t+1}).
\]
Again, the fact that
$s_i(k',A_i^{t+1} \setminus \{c\})=0$
for every $k' < k$,
implies the inductive thesis.

\textbf{\underline{Part 2: The output allocation is \PO{}.}}

The proof of \PO{} is similar to
the respective part in the proof of \cref{thm:goods:efx}.
We will prove that
for every $t \in \{0,\dots,m\}$,
allocation $A^t$ is PO.
Assume otherwise and take the first $t$
such that $A^t$ is not \PO{}.
By \cref{thm:po:ebadian},
this means that there is an alternating path,
$p = (c_0, i_1, c_1, \dots, i_s, c_s)$
and agent $i_0$ such that $(c_s, i_0) \in A^t$
and $\psi(i_0, c_0) > \psi(c_s,i_0)$.
Since $A^{t-1}$ is \PO{},
either $(c_s, i_0) \not \in A^{t-1}$ or
$p$ is not an alternating path for $A^{t-1}$.
In both cases,
at least one chore on path $p$
belongs to different agents in
$A^t$ and $A^{t-1}$.
Let us take the first on the path
such chore $c_u$, i.e.,
such that $u \in [s]$ is the smallest.
Observe that $c_u$ is available in $t-1$.
Also, the part of path $p$ that is 
$(c_0, i_1, c_1, \dots, i_u, c_u)$
is still an alternating path in $A^{t-1}$
as chores $c_0,\dots,c_{u-1}$
belong to the same agents in $A^t$ and $A^{t-1}$.
Both facts imply, that $c_0$ is available in $t-1$.
Hence, by \cref{claim:chores:ef1:best},
agent $i_0$ weakly prefers
each of its chores in $A^t$ over $c_0$.
However, this contradicts the fact that
$\psi(i_0, c_0) > \psi(c_s,i_0)$
and $(c_s, i_0) \in A^t$.
\end{proof}

\section{PO Characterization for Chores}
\label{app:chores:po}

In this appendix, we present the proof of the \PO{} characterization for weakly lexicographic chores only instances.

\thmchorespo*

\begin{proof}
    If there exists such alternating path
    and agent $i_0$, 
    then let us obtain $A'$ from $A$ 
    by updating it along this path, i.e., 
    \(
        A' =  A \setminus \{(c_s, i_0)\} \setminus
            \{(c_{r-1}, i_r) : r \in [s]\} \cup
            \{(i_r, c_r) : r \in \{0, \dots, s\}\}.
    \)
    From the definition of alternating path,
    for every $r \in [s]$,
    agent $i_r$ weakly prefers
    its bundle in $A'$ over its bundle in $A$.
    Moreover, $i_0$ strictly prefers its bundle in $A'$.
    Thus, we get that $A'$ Pareto dominates $A$.

    Thus, it remains to show that if $A$ is not \PO{},
    then there exists an alternating path 
    and agent $i_0$ in question.
    Take an arbitrary $A'$ that Pareto dominates $A$
    and for which the number of reallocated chores, i.e.,
    \(
        \sum_{i \in N} |A_i \setminus A'_i|,
    \)
    is minimal.
    By $N' = \{ i \in N : A'_i \neq A_i\}$
    let us denote
    the set of agents with different bundles in $A'$ and $A$.

    For every agent $i \in N'$,
    by $x_i$ let us denote the first position
    such that $s_i(x_i, A_i \setminus A'_i)>0$
    (observe that there always exists such $x_i$
    since $i$ cannot just receive chores without
    giving away any).
    For every chore $c \in M$ and agent $i \in N'$ such that 
    $c \in A_i \setminus A'_i$, let us use $b_c(i)$ to denote
    agent $j$ that has $c$ in $A'$, i.e.,  
    $c \in A'_j \setminus A_j$. 
    
    Next, let us consider a directed multigraph, $G = (N',E)$
    in which the set of vertices is $N'$
    and for every agent $i \in N'$
    and chore $c \in \rhd(x_i, A_i \setminus A'_i)$
    we put an edge $(i, b_c(i))$.
    Intuitively, graph $G$ denotes who has to give the chores to whom to convert $A$ to $A'$
    (where we take into account only the most important chores from the perspective of the giver).

    Observe that every vertex in $G$ has an outgoing edge.
    Hence, there has to be a cycle, $C = (i_0,i_1,i_2,\dots,i_s)$.
    For every $r \in [s]$,
    let $c_{r-1}$ be a chore that is "given" by $i_r$ to $i_{r-1}$
    and $c_s$ by $i_0$ to $i_s$, i.e., 
    $b_{c_{r-1}}(i_r) = i_{r-1}$ and $b_{c_s}(i_0) = i_s$.
    Fix $r \in [s]$.
    Since $c_{r-1}$ is in the $x_{i_r}$-th indifference class of agent $i_r$
    chore $c_{r}$ has to be in the class $x_{i_r}$ or further
    otherwise agent $i_r$
    would strictly prefer its bundle in $A$ over $A'$.
    Similarly, chore $c_{s}$ has to be in the class $x_{i_0}$ or further for agent $i_0$.
    If for every $r \in \{0,\dots,s\}$
    it would be in exactly $x_{i_r}$-th class, i.e., not higher,
    then consider $A''$ obtained from $A'$, i.e.,
    \(
        A'' = A' \setminus
        \{(c_{r-1}, i_{r}) : r \in [s]\} \setminus
        \{(c_s,i_0)\} \cup
        \{(c_{r}, i_r) : r \in \{0, \dots,s\} \}.
    \)
    Observe that every agent is indifferent between the bundles in $A''$ and $A'$.
    Hence, $A''$ also Pareto dominates $A$.
    However, 
    \(
        \sum_{i \in N} |A_i \setminus A''_i| <
        \sum_{i \in N} |A_i \setminus A'_i|,
    \)
    which is a contradiction.
    Thus, there exists an agent $i_r$ in the cycle $C$ that strictly prefers the new chore to the old one.
    Without loss of generality, let us assume that $r=0$, i.e.,
    $c_0$ is in the later indifference class of $i_0$ than $x_{i_0}$.
    Then, observe that $(c_0,i_1,c_1,\dots,i_s,c_s)$
    is an alternating path for allocation $A$ such that 
    $(c_s,i_0) \in A$
    and $\psi(i_0,c_0) > \psi(c_s,i_0)$,
    which concludes the proof.
\end{proof}


\begin{thebibliography}{46}
\providecommand{\natexlab}[1]{#1}
\providecommand{\url}[1]{\texttt{#1}}
\expandafter\ifx\csname urlstyle\endcsname\relax
  \providecommand{\doi}[1]{doi: #1}\else
  \providecommand{\doi}{doi: \begingroup \urlstyle{rm}\Url}\fi

\bibitem[Amanatidis et~al.(2021)Amanatidis, Birmpas, Filos-Ratsikas, Hollender, and Voudouris]{ABF+21maximum}
Georgios Amanatidis, Georgios Birmpas, Aris Filos-Ratsikas, Alexandros Hollender, and Alexandros~A Voudouris.
\newblock Maximum {N}ash welfare and other stories about {EFX}.
\newblock \emph{Theoretical Computer Science}, 863:\penalty0 69--85, 2021.

\bibitem[Aziz and De~Keijzer(2012)]{aziz2012housing}
Haris Aziz and Bart De~Keijzer.
\newblock Housing markets with indifferences: A tale of two mechanisms.
\newblock In \emph{Proceedings of the 26th AAAI conference on Artificial Intelligence}, pages 1249--1255, 2012.

\bibitem[Aziz et~al.(2015)Aziz, Gaspers, Mackenzie, and Walsh]{aziz2015fair}
Haris Aziz, Serge Gaspers, Simon Mackenzie, and Toby Walsh.
\newblock Fair assignment of indivisible objects under ordinal preferences.
\newblock \emph{Artificial Intelligence}, 227:\penalty0 71--92, 2015.

\bibitem[Aziz et~al.(2017)Aziz, Rauchecker, Schryen, and Walsh]{aziz2017algorithms}
Haris Aziz, Gerhard Rauchecker, Guido Schryen, and Toby Walsh.
\newblock Algorithms for max-min share fair allocation of indivisible chores.
\newblock In \emph{Proceedings of the 31st AAAI Conference on Artificial Intelligence}, pages 335--341, 2017.

\bibitem[Aziz et~al.(2019)Aziz, Bir{\'{o}}, Lang, Lesca, and Monnot]{aziz2019reallocation}
Haris Aziz, P{\'{e}}ter Bir{\'{o}}, J{\'{e}}r{\^{o}}me Lang, Julien Lesca, and J{\'{e}}r{\^{o}}me Monnot.
\newblock Efficient reallocation under additive and responsive preferences.
\newblock \emph{Theoretical Computer Science}, 790:\penalty0 1--15, oct 2019.
\newblock \doi{10.1016/j.tcs.2019.05.011}.

\bibitem[Aziz et~al.(2022)Aziz, Caragiannis, Igarashi, and Walsh]{aziz2022fair}
Haris Aziz, Ioannis Caragiannis, Ayumi Igarashi, and Toby Walsh.
\newblock Fair allocation of indivisible goods and chores.
\newblock \emph{Autonomous Agents and Multi-Agent Systems}, 36\penalty0 (1):\penalty0 1--21, 2022.

\bibitem[Aziz et~al.(2023{\natexlab{a}})Aziz, Li, Xing, and Zhou]{aziz2023possible}
Haris Aziz, Bo~Li, Shiji Xing, and Yu~Zhou.
\newblock Possible fairness for allocating indivisible resources.
\newblock In \emph{Proceedings of the 22nd International Conference on Autonomous Agents and Multiagent Systems}, pages 197--205, 2023{\natexlab{a}}.

\bibitem[Aziz et~al.(2023{\natexlab{b}})Aziz, Lindsay, Ritossa, and Suzuki]{aziz2023fairtypes}
Haris Aziz, Jeremy Lindsay, Angus Ritossa, and Mashbat Suzuki.
\newblock Fair allocation of two types of chores.
\newblock In \emph{Proceedings of the 22nd International Conference on Autonomous Agents and Multiagent Systems}, pages 143--151, 2023{\natexlab{b}}.

\bibitem[Babaioff et~al.(2019)Babaioff, Nisan, and Talgam-Cohen]{babaioff2019fair}
Moshe Babaioff, Noam Nisan, and Inbal Talgam-Cohen.
\newblock Fair allocation through competitive equilibrium from generic incomes.
\newblock In \emph{Proceedings of the 2nd ACM Conference on Fairness, Accountability and Transparency}, pages 180--180, 2019.

\bibitem[Babaioff et~al.(2021)Babaioff, Ezra, and Feige]{BEF21fair}
Moshe Babaioff, Tomer Ezra, and Uriel Feige.
\newblock Fair and truthful mechanisms for dichotomous valuations.
\newblock In \emph{Proceedings of the 35th AAAI Conference on Artificial Intelligence}, volume~35, pages 5119--5126, 2021.

\bibitem[Barman et~al.(2018)Barman, Krishnamurthy, and Vaish]{barman2018finding}
Siddharth Barman, Sanath~Kumar Krishnamurthy, and Rohit Vaish.
\newblock Finding fair and efficient allocations.
\newblock In \emph{Proceedings of the 19th ACM Conference on Economics and Computation}, pages 557--574, 2018.

\bibitem[Bhaskar et~al.(2021)Bhaskar, Sricharan, and Vaish]{bhaskar2021approximate}
Umang Bhaskar, AR~Sricharan, and Rohit Vaish.
\newblock On approximate envy-freeness for indivisible chores and mixed resources.
\newblock \emph{Approximation, Randomization, and Combinatorial Optimization. Algorithms and Techniques}, 2021.

\bibitem[Bogomolnaia et~al.(2005)Bogomolnaia, Deb, and Ehlers]{bogomolnaia2005strategy}
Anna Bogomolnaia, Rajat Deb, and Lars Ehlers.
\newblock Strategy-proof assignment on the full preference domain.
\newblock \emph{Journal of Economic Theory}, 123\penalty0 (2):\penalty0 161--186, 2005.

\bibitem[Budish(2011)]{budish2011combinatorial}
Eric Budish.
\newblock The combinatorial assignment problem: Approximate competitive equilibrium from equal incomes.
\newblock \emph{Journal of Political Economy}, 119\penalty0 (6):\penalty0 1061--1103, 2011.

\bibitem[Camacho et~al.(2023)Camacho, Fonseca-Delgado, P{\'e}rez, and Tapia]{camacho2023generalized}
Franklin Camacho, Rigoberto Fonseca-Delgado, Ram{\'o}n~Pino P{\'e}rez, and Guido Tapia.
\newblock Generalized binary utility functions and fair allocations.
\newblock \emph{Mathematical Social Sciences}, 121:\penalty0 50--60, 2023.

\bibitem[Caragiannis et~al.(2019)Caragiannis, Kurokawa, Moulin, Procaccia, Shah, and Wang]{caragiannis2019unreasonable}
Ioannis Caragiannis, David Kurokawa, Herv{\'e} Moulin, Ariel~D Procaccia, Nisarg Shah, and Junxing Wang.
\newblock The unreasonable fairness of maximum {N}ash welfare.
\newblock \emph{ACM Transactions on Economics and Computation (TEAC)}, 7\penalty0 (3):\penalty0 1--32, 2019.

\bibitem[Ebadian et~al.(2022)Ebadian, Peters, and Shah]{ebadian2022fairly}
Soroush Ebadian, Dominik Peters, and Nisarg Shah.
\newblock How to fairly allocate easy and difficult chores.
\newblock In \emph{Proceedings of the 21st International Conference on Autonomous Agents and MultiAgent Systems}, pages 372--380, 2022.

\bibitem[Fishburn(1974)]{fishburn1974exceptional}
Peter~C Fishburn.
\newblock Lexicographic orders, utilities and decision rules: A survey.
\newblock \emph{Management science}, 20\penalty0 (11):\penalty0 1442--1471, 1974.

\bibitem[Fujita et~al.(2018)Fujita, Lesca, Sonoda, Todo, and Yokoo]{fujita2018complexity}
Etsushi Fujita, Julien Lesca, Akihisa Sonoda, Taiki Todo, and Makoto Yokoo.
\newblock A complexity approach for core-selecting exchange under conditionally lexicographic preferences.
\newblock \emph{Journal of Artificial Intelligence Research}, 63:\penalty0 515--555, 2018.

\bibitem[Furma{\'n}czyk et~al.(2013)Furma{\'n}czyk, Kaliraj, Kubale, and Vivin]{furmanczyk2013equitable}
Hanna Furma{\'n}czyk, K~Kaliraj, Marek Kubale, and J~Vernold Vivin.
\newblock Equitable coloring of corona products of graphs.
\newblock \emph{Advances and Applications of Discrete Mathematics}, 11:\penalty0 103--120, 2013.

\bibitem[Gafni et~al.(2021)Gafni, Huang, Lavi, and Talgam-Cohen]{gafni2021unified}
Yotam Gafni, Xin Huang, Ron Lavi, and Inbal Talgam-Cohen.
\newblock Unified fair allocation of goods and chores via copies.
\newblock \emph{arXiv preprint arXiv:2109.08671}, 2021.

\bibitem[Garg and Murhekar(2021)]{garg2021computing}
Jugal Garg and Aniket Murhekar.
\newblock Computing fair and efficient allocations with few utility values.
\newblock In \emph{Proceedings of the 14th International Symposium on Algorithmic Game Theory}, pages 345--359. Springer, 2021.

\bibitem[Garg and Taki(2021)]{garg2020improved}
Jugal Garg and Setareh Taki.
\newblock An improved approximation algorithm for maximin shares.
\newblock \emph{Artificial Intelligence}, 300:\penalty0 103547, 2021.

\bibitem[Garg et~al.(2022)Garg, Murhekar, and Qin]{garg2022fair}
Jugal Garg, Aniket Murhekar, and John Qin.
\newblock Fair and efficient allocations of chores under bivalued preferences.
\newblock \emph{Proceedings of 36th the AAAI Conference on Artificial Intelligence}, pages 5043--5050, 2022.

\bibitem[Garg et~al.(2010)Garg, Kavitha, Kumar, Mehlhorn, and Mestre]{garg2010assigning}
Naveen Garg, Telikepalli Kavitha, Amit Kumar, Kurt Mehlhorn, and Juli{\'a}n Mestre.
\newblock Assigning papers to referees.
\newblock \emph{Algorithmica}, 58:\penalty0 119--136, 2010.

\bibitem[Gary and Johnson(1979)]{gary1979computers}
Michael~R Gary and David~S Johnson.
\newblock Computers and intractability: A guide to the theory of {NP}-completeness, 1979.

\bibitem[Ghodsi et~al.(2021)Ghodsi, Hajiaghayi, Seddighin, Seddighin, and Yami]{ghodsi2018fair}
Mohammad Ghodsi, Mohammad~Taghi Hajiaghayi, Masoud Seddighin, Saeed Seddighin, and Hadi Yami.
\newblock Fair allocation of indivisible goods: Improvement.
\newblock \emph{Mathematics of Operations Research}, 46\penalty0 (3):\penalty0 1038--1053, 2021.

\bibitem[Gigerenzer and Goldstein(1996)]{gigerenzer1996reasoning}
Gerd Gigerenzer and Daniel~G Goldstein.
\newblock Reasoning the fast and frugal way: models of bounded rationality.
\newblock \emph{Psychological review}, 103\penalty0 (4):\penalty0 650, 1996.

\bibitem[Hosseini and Larson(2019)]{HL19multiple}
Hadi Hosseini and Kate Larson.
\newblock {Multiple Assignment Problems under Lexicographic Preferences}.
\newblock In \emph{Proceedings of the 18th International Conference on Autonomous Agents and MultiAgent Systems}, pages 837--845, 2019.

\bibitem[Hosseini et~al.(2020)Hosseini, Sikdar, Vaish, Wang, and Xia]{hosseini2020HEF}
Hadi Hosseini, Sujoy Sikdar, Rohit Vaish, Hejun Wang, and Lirong Xia.
\newblock Fair division through information withholding.
\newblock In \emph{Proceedings of the 34th AAAI Conference on Artificial Intelligence}, pages 2014--2021, 2020.

\bibitem[Hosseini et~al.(2021)Hosseini, Sikdar, Vaish, and Xia]{hosseini2021fair}
Hadi Hosseini, Sujoy Sikdar, Rohit Vaish, and Lirong Xia.
\newblock Fair and efficient allocations under lexicographic preferences.
\newblock In \emph{Proceedings of the 35th AAAI Conference on Artificial Intelligence}, pages 5472--5480, 2021.

\bibitem[Hosseini et~al.(2022{\natexlab{a}})Hosseini, Searns, and Segal-Halevi]{hosseini2022ordinal}
Hadi Hosseini, Andrew Searns, and Erel Segal-Halevi.
\newblock Ordinal maximin share approximation for goods.
\newblock \emph{Journal of Artificial Intelligence Research}, 74:\penalty0 353--391, 2022{\natexlab{a}}.

\bibitem[Hosseini et~al.(2022{\natexlab{b}})Hosseini, Sikdar, Vaish, and Xia]{hosseini2022fairly}
Hadi Hosseini, Sujoy Sikdar, Rohit Vaish, and Lirong Xia.
\newblock Fairly dividing mixtures of goods and chores under lexicographic preferences.
\newblock In \emph{Proceedings of the 22nd International Conference on Autonomous Agents and MultiAgent Systems}, 2022{\natexlab{b}}.

\bibitem[Hosseini et~al.(2023)Hosseini, Mammadov, and W{\k{a}}s]{hosseini2023fairly}
Hadi Hosseini, Aghaheybat Mammadov, and Tomasz W{\k{a}}s.
\newblock Fairly allocating goods and (terrible) chores.
\newblock In \emph{Proceedings of the 32nd International Joint Conference on Artificial Intelligence}, pages 2738--2746, 2023.

\bibitem[Jaramillo and Manjunath(2012)]{jaramillo2012difference}
Paula Jaramillo and Vikram Manjunath.
\newblock The difference indifference makes in strategy-proof allocation of objects.
\newblock \emph{Journal of Economic Theory}, 147\penalty0 (5):\penalty0 1913--1946, 2012.

\bibitem[Katta and Sethuraman(2006)]{katta2006solution}
Akshay-Kumar Katta and Jay Sethuraman.
\newblock A solution to the random assignment problem on the full preference domain.
\newblock \emph{Journal of Economic Theory}, 131\penalty0 (1):\penalty0 231--250, 2006.

\bibitem[Klaus and Meo(2023)]{klaus2023core}
Bettina Klaus and Claudia Meo.
\newblock The core for housing markets with limited externalities.
\newblock \emph{Economic Theory}, pages 1--33, 2023.

\bibitem[Krysta et~al.(2014)Krysta, Manlove, Rastegari, and Zhang]{krysta2014size}
Piotr Krysta, David Manlove, Baharak Rastegari, and Jinshan Zhang.
\newblock Size versus truthfulness in the house allocation problem.
\newblock In \emph{Proceedings of the 15th ACM conference on Economics and Computation}, pages 453--470, 2014.

\bibitem[Kurokawa et~al.(2018)Kurokawa, Procaccia, and Wang]{kurokawa2018fair}
David Kurokawa, Ariel~D Procaccia, and Junxing Wang.
\newblock Fair enough: Guaranteeing approximate maximin shares.
\newblock \emph{Journal of the ACM (JACM)}, 65\penalty0 (2):\penalty0 1--27, 2018.

\bibitem[Lang et~al.(2018)Lang, Mengin, and Xia]{lang2018voting}
J{\'e}r{\^o}me Lang, J{\'e}r{\^o}me Mengin, and Lirong Xia.
\newblock Voting on multi-issue domains with conditionally lexicographic preferences.
\newblock \emph{Artificial Intelligence}, 265:\penalty0 18--44, 2018.

\bibitem[Li et~al.(2022)Li, Li, and Wu]{li2022almost}
Bo~Li, Yingkai Li, and Xiaowei Wu.
\newblock Almost (weighted) proportional allocations for indivisible chores.
\newblock In \emph{Proceedings of the ACM Web Conference 2022}, pages 122--131, 2022.

\bibitem[Lipton et~al.(2004)Lipton, Markakis, Mossel, and Saberi]{lipton2004approximately}
Richard~J Lipton, Evangelos Markakis, Elchanan Mossel, and Amin Saberi.
\newblock On approximately fair allocations of indivisible goods.
\newblock In \emph{Proceedings of the 5th ACM Conference on Electronic Commerce}, pages 125--131, 2004.

\bibitem[Nguyen(2020)]{nguyen2020fairly}
Trung~Thanh Nguyen.
\newblock How to fairly allocate indivisible resources among agents having lexicographic subadditive utilities.
\newblock In \emph{Frontiers in Intelligent Computing: Theory and Applications}, pages 156--166. Springer, 2020.

\bibitem[Plaut and Roughgarden(2020)]{PR20almost}
Benjamin Plaut and Tim Roughgarden.
\newblock Almost envy-freeness with general valuations.
\newblock \emph{SIAM Journal on Discrete Mathematics}, 34\penalty0 (2):\penalty0 1039--1068, 2020.

\bibitem[Saban and Sethuraman(2013)]{saban2013house}
Daniela Saban and Jay Sethuraman.
\newblock House allocation with indifferences: a generalization and a unified view.
\newblock In \emph{Proceedings of the 14th ACM conference on Electronic Commerce}, pages 803--820, 2013.

\bibitem[Saban and Sethuraman(2014)]{saban2014note}
Daniela Saban and Jay Sethuraman.
\newblock A note on object allocation under lexicographic preferences.
\newblock \emph{Journal of Mathematical Economics}, 50:\penalty0 283--289, 2014.

\end{thebibliography}
\end{document}